\documentclass[sigconf,balance=false,authorversion]{acmart}

\makeatletter                   
\def\mdseries@tt{m}             
\makeatother                    
\usepackage[plain]{fancyref}
\usepackage[draft=true]{minted} 
\usepackage{color}
\usepackage{hyperref}           
\hypersetup{
    colorlinks=true,
    linkcolor=blue,
    filecolor=red,      
    urlcolor=magenta,
    breaklinks=true,            
}
\usepackage{breakurl}           

\usepackage{algorithm}
\usepackage{algorithmicx}
\usepackage{algpseudocode}
\usepackage{textcomp}

\def\BibTeX{{\rm B\kern-.05em{\sc i\kern-.025em b}\kern-.08emT\kern-.1667em\lower.7ex\hbox{E}\kern-.125emX}}
    
\acmYear{2019}
\acmConference[Submission for Review]{}
\acmBooktitle{Journal or Conference}

\begin{document}
\sloppy                         
%
\title[Aggregating Votes with Local Differential Privacy]{Aggregating Votes with Local Differential Privacy: \\Usefulness, Soundness vs. Indistinguishability}

%
\author{Shaowei Wang, Jiachun Du}
\email{seawellwang@tencent.com}
\affiliation{
  \institution{Tencent Games}
}
\author{Wei Yang$^*$, Xinrong Diao, Zichun Liu,\ \ \ \ \ \ \ \  Yiwen Nie, Liusheng Huang$^*$, Hongli Xu}
\affiliation{University of Science and Technology of China\\
$^*$ Corresponding Authors}

%
\renewcommand{\shortauthors}{}

%
\begin{abstract}
Voting plays a central role in bringing crowd wisdom to collective decision making, meanwhile data privacy has been a common ethical/legal issue in eliciting preferences from individuals. This work studies the problem of aggregating individual's voting data under the local differential privacy setting, where usefulness and soundness of the aggregated scores are of major concern. One naive approach to the problem is adding Laplace random noises, however, it makes aggregated scores extremely fragile to new types of strategic behaviors tailored to the local privacy setting: data amplification attack and view disguise attack. The data amplification attack means an attacker's manipulation power is amplified by the privacy-preserving procedure when contributing a fraud vote. The view disguise attack happens when an attacker could disguise malicious data as valid private views to manipulate the voting result. 

In this work, after theoretically quantifying the estimation error bound and the manipulating risk bound of the Laplace mechanism, we propose two mechanisms improving the usefulness and soundness simultaneously: the weighted sampling mechanism and the additive mechanism. The former one interprets the score vector as probabilistic data. Compared to the Laplace mechanism for Borda voting rule with $d$ candidates, it reduces the mean squared error bound by half and lowers the maximum magnitude risk bound from $+\infty$ to $O(\frac{d^3}{n\epsilon})$. The latter one randomly outputs a subset of candidates according to their total scores. Its mean squared error bound is optimized from $O(\frac{d^5}{n\epsilon^2})$ to $O(\frac{d^4}{n\epsilon^2})$, and its maximum magnitude risk bound is reduced to $O(\frac{d^2}{n\epsilon})$. Experimental results validate that our proposed approaches averagely reduce estimation error by $50\%$ and are more robust to adversarial attacks.

\end{abstract}

%
%

\ccsdesc[500]{Security and privacy~Privacy protections}

%
\keywords{privacy, vote, security, aggregation, election}

%

%
\maketitle

\section{Introduction}\label{sec:intro}
Collective decision making is widely adopted by governing organizations and commercial service providers, which benefits from the wisdom of crowd via aggregating individual preferences. For example, during an election for social choice, a profile of ranking data from voters is summarized to determine the final preference ordering over several options; when online service providers are choosing the next movement, millions of users' preference data are aggregated to measure relative popularity between alternative treatments.

The mapping from many individual preferences to a single resulting ordering is called the voting rule. An intuitive and fundamental class of voting rule is positional voting rule, the general idea of which is assigning each candidate a score according to its position/rank in each voter's preference. Examples of positional voting rule include the Borda counts \cite{black1958theory}, plurality voting and Nauru voting \cite{reilly2002social}. Specifically in Borda voting, the $i$-th candidate in a vote scores $d-i$ points, where $d$ is the number of options or candidates. If $4$ voters' preferences over $5$ candidates are:
$$\text{voter}\ 1:\ \ \ A_3\ \ \ A_2\ \ \ A_1\ \ \ A_4\ \ \ A_5;\ \ \ \ \ \ \ \ \ \ \text{voter}\ 2:\ \ \ A_2\ \ \ A_3\ \ \ A_5\ \ \ A_4\ \ \ A_1;$$
$$\text{voter}\ 3:\ \ \ A_5\ \ \ A_2\ \ \ A_3\ \ \ A_4\ \ \ A_1;\ \ \ \ \ \ \ \ \ \ \text{voter}\ 4:\ \ \ A_1\ \ \ A_2\ \ \ A_5\ \ \ A_3\ \ \ A_4,$$
following the rule of Borda counts, the candidate $A_3$ scores $10$ points, and the candidate $A_2$ wins the voting with highest $13$ points. 


Privacy is a basic requirement in secured voting systems \cite{fujioka1992practical}, as providing secrecy of votes could avoid leaking personal preferences and help to elicit honest responses, especially when voting on sensitive topics. Privacy threats not only come from parties outside the voting system, but also come from the administrator, the counter or other voters who might want to infer certain voters' votes. Secure multi-party computation \cite{chaum1988multiparty} alleviates these problems by securely aggregating scores, but is still fragile to collusion between the counter and other voters, and may have efficiency issues for large-scale online voting systems involving millions of voters. Differential privacy \cite{dwork2011differential} could also be employed for privacy preserving voting, which ensures distinguishability of results no matter any single vote presents or not in the voting profile. However, it relies on the existence of a trustful data curator/counter for all voters. 

Another paradigm to privacy preserving voting is local differential privacy, which sanitizes the vote locally and independently on the voter's side, and ensures up to $\exp(\epsilon)$ distinguishability on outputting probabilities no matter what the true vote a voter holds. Local privacy has the advantages of being information-theoretically rigid, computationally efficient and operationally flexible. The voter has full controllability during the privacy preserving procedure without the trust of any parties, the counter/administrator is also tolerable to voters' unsynchronized opt-out, withdrawal and modification actions on votes. These advantages make local differential privacy the best fit for recently enacted General Data Protection Regulation (GDPR \cite{voigt2017eu}) in the European Union, which emphasizes data owner's controllability on contribution, storage, analysis and transfer of their data. 


A straight-forward way to realize local differential privacy for voting data is adding Laplace random noises. After representing votes in a score form as in Table \ref{tab:scoredvote}, Laplace noises of scale $\frac{\epsilon}{\Delta}$ are independently added to each score in a vote, here the $\epsilon$ is privacy level and the $\Delta$ is the maximum absolute difference between any two scored votes. If $d$ is odd, we have $\Delta=\frac{d^2-1}{2}$ for Borda voting. For example, the scored vote $v^{(1)}$ possibly becomes:
$$\tilde{v}^{(1)}=[26.6,\ -45.2,\  6.3,\ -7.3,\ -1.5],$$
after adding Laplace noises of scale $\frac{1.0}{12}$. In order to reach the criterion of local differential privacy, the (unbiased) private view $\tilde{v}$ might far deviate from the true vote $v$. The expected deviation of the private view here is $\mathbb{E}[|\tilde{v}-v|_2^2]=\frac{2d\Delta^2}{\epsilon^2}$.  
This formula indicates preserving privacy comes at the cost of data usefulness, improving which is the central focus in the current local privacy literature  (e.g., in \cite{duchi2013local,kairouz2014extremal,kairouz2016discrete,bun2018heavy}).    
    
\begin{table}[H]
	\setlength{\tabcolsep}{1.0em}
	\renewcommand{\arraystretch}{1.0}
	\caption{An example of $4$ scored votes on $5$ candidates in Borda voting.}
	\label{tab:scoredvote}
	\centering
	\begin{tabular}{c|ccccc}
		 & $A_1$ & $A_2$ & $A_3$ & $A_4$ & $A_5$\\
		\hline
		$v^{(1)}$ & 2 & 3 & 4 & 1 & 0\\
		$v^{(2)}$ & 0 & 4 & 3 & 1 & 2\\
		$v^{(3)}$ & 0 & 3 & 2 & 1 & 4 \\
		$v^{(4)}$ & 4 & 3 & 1 & 0 & 2\\
		\hline
		total score & 6 & 13 & 10 & 3 & 8\\
		average score & 1.5 & 3.25 & 2.5 & 0.75 & 2.0\\
	\end{tabular}
\end{table}

\subsection{Privacy Induced Attacks}

For voting or other data aggregation systems whose results have critical consequences, the soundness of the system against to strategic participants is also fundamental. We observe that local privacy preservation comes at the cost of also the soundness. The power of adversarial attacks (e.g., vote fraud) might be amplified by the privacy preserving procedure, disguising as private views also makes adversaries much easier to manipulate the voting result. 

\textbf{Data amplification attack.}\ \ 
In the case that an adversary could contribute fraud data but is unable to skip the privacy preserving procedure, the effects that the fraud data can have on the result might be amplified due to the intrinsic randomness nature of privacy preservation. A more rigid level of privacy preservation means the private view will have more randomness, and hence have more (maximum or expected) magnitude. Take the Laplace approach as an example, in the non-private setting (e.g., $\epsilon=+\infty$), the magnitude of one vote is $|v|=\sum_{j\in [1,d]}(d-j)=\frac{d(d-1)}{2}=10$,
but the maximum magnitude of a private view $\tilde{v}$ becomes infinite as Laplace random noises are unbounded. The expected magnitude of the private view is a function of privacy level as follows: 
$$\mathbb{E}[|\tilde{v}|] = \frac{\Delta}{\epsilon}\cdot\frac{e^{\epsilon/\Delta}-e^{(1-d)\epsilon/\Delta}}{e^{\epsilon/\Delta}-1}+\frac{d(d-1)}{2}.$$
Although local differential privacy ensures indistinguishability on outputs of any possible votes, hence the constructive power of one vote on the voting result is diminished,  but its deconstructive power is amplified.           
     
\textbf{View disguise attack.}\ \ 
In the case that an adversary has direct control on the private view sent to the aggregator, the adversary will be able to disguise a malicious private view as an ordinary (randomized) one, and thus make constructive/deconstructive changes to the aggregation result. The domain of private views is broader than the true data and grows with the level of privacy, hence an adversary's constructive/deconstructive power becomes larger. For example, in the non-private setting, the domain of a scored vote is bounded by $[0,d-1]^d$, while in the Laplace approach, the domain of the private view is scaled to $[-\infty,+\infty]^d$. Even though we can filter out private views that are extremely unlikely to be observed, the filtered domain $[-\tilde{\Theta}(\frac{1}{\epsilon}),+\tilde{\Theta}(\frac{1}{\epsilon})]$ also grows with the level of privacy preservation (see Section \ref{subsec:soundlaplace} for detail). Compared to selecting a value from the domain of scored votes in the non-private setting,  an adversary is easier to manipulate the election result by selecting a value from the (filtered) domain of private views.

\subsection{Our Contributions}
As a remedy to above issues in the naive approach to local differential private vote aggregation, this work aims to develop novel mechanisms improving the usefulness and soundness. The main content and contributions of this work are summarized as follows:
\begin{itemize}
\item[I.]{We identify soundness issues due to privacy preservation in local private data aggregation systems and categorize them as data amplification/view disguise attack based on an adversary's controllability over the privacy preserving procedure. The data amplification attack captures an adversary's deconstructive power on the aggregation result by contributing fraudulent data. The view disguise attack captures an adversary's constructive/deconstructive power on the aggregation result by directly disguising private views. Formal quantified metrics are defined (in Section \ref{sec:model}) to measure the power of adversarial attacks and the soundness of local private voting systems.}

\item[II.]{We present thorough analyses of Laplace mechanism for local private vote aggregation (in Section \ref{sec:laplace}, partially in former paragraphs), including sensitivity bounds of general positional voting rules, error bounds of estimated votes, and risk bounds under strategic attacks.}

\item[III.]{A novel alternative to the Laplace approach: the weighted sampling mechanism, is proposed for general positional voting rules  (in Section \ref{sec:sampling}). The mechanism samples an option with a probability mass proportional to its score and then applies well-studied local private methods on the categorical option. For Borda counts, this mechanism reduces estimation error bound from $\sim\frac{d^5}{2n\epsilon^2}$ to $\sim\frac{d^5}{4n\epsilon^2}$, and reduces the maximum manipulation risk bound from $+\infty$ to $O(\frac{d^3}{n\epsilon})$, compared to the Laplace mechanism.}



\item[IV.]{ Given that the weighted sampling mechanism works unsatisfactorily in the high privacy regime, we further propose the additive mechanism (in Section \ref{sec:additive}), which samples a subset of candidates according to the summation of their scores. The sampling problem underlying the mechanism is a strict case of the weighted random sampling problem \cite{efraimidis2006weighted}, we provide a recursive algorithm as a solution. The additive mechanism has estimation error bound of $O(\frac{d^4}{n\epsilon^2})$, and expected/maximum manipulation risk bounds both at $O(\frac{d^2}{n\epsilon})$.}

\item[V.]{We discuss the interaction/trade-off between usefulness, soundness, truthfulness and indistinguishability in local private data aggregation (in Section \ref{sec:discussion}). Quantified relation between estimation error bound and manipulation risk bound is built, which shows optimizing usefulness usually benefits soundness.  
}
\item[ VI.]{Experiments on extensive voting scenarios are conducted (i.e., in Section \ref{sec:exp}) to validate the usefulness and soundness of proposed mechanisms. Results demonstrate that estimation errors decrease by $1/2$ and manipulation risks are also significantly reduced, when compared to the existing approaches.}

\end{itemize}    

\section{Related Work}\label{sec:related}
Security requirements in voting systems cover many aspects,  such as privacy, verifiability and soundness. Here we review some representative works on the privacy/anonymity/soundness in the area of electronic voting and computational social choice, then retrospect recent works on privacy preserving data analyses within the differential privacy framework.

\subsection{Security in Voting Systems}
\subsubsection{Privacy/Anonymity}
Since the seminal work of Chaum \cite{chaum1981untraceable}, plenty of cryptographic schemes have contributed to keeping voters or (and) votes secret in electronic voting systems. Schemes based on homomorphic encryption operate on encrypted votes to compute the sum/average of votes without knowing plain-text of votes (e.g., ElGamal encryption in \cite{mu1998anonymous,hirt2010receipt}, Paillier encryption in \cite{ryan2008pret,xia2008analysis}, and \cite{hirt2000efficient,peng2004multiplicative}), hence keeps votes private from the vote counter or adversaries. Further combining cryptographic techniques with anonymous channels (e.g, the mixnet in \cite{park1993efficient,abe1998universally,lee2003providing,bulens2011running}) that randomly shuffle a bundle of messages from voters, votes (or ciphertexts) are then unlinkable to source voters. Having one single vote counter in a voting system is intolerant to failures or attacks, several works hence then multiple authorities secret sharing on decryption key \cite{cramer1997secure} to improve robustness of the system. In case of multiple authorities, the secrecy of every vote could be improved by decomposing them into several parts, each part is then sent to some authorities (e.g, in \cite{cohen1985robust,benaloh1986distributing,cramer1996multi}). Consequently, corrupted authorities less than a threshold number are unable to derive complete information of a vote.

Another line of works that could be employed for privacy preserving voting systems is data perturbation, which uses techniques of generalization (e.g., $k$-anonymity \cite{sweeney2002k} in \cite{chatzikokolakis2008anonymity,zhao2016voting}) and randomization (e.g., Gaussian noise adding \cite{dwork2006our}, randomized response \cite{warner1965randomized} and differential privacy \cite{dwork2011differential}) to hide the exact values of each vote and the voting result. Compared to cryptographic techniques provide computational secrecy and anonymity of votes/voters, data perturbation approaches are usually much more efficient for implementations. Among them, classic privacy notions and techniques like $k$-anonymity and Gaussian noise adding have shown to be risky for adversaries with prior knowledge \cite{kargupta2003privacy,machanavajjhala2006diversity,li2007t}.


\subsubsection{Soundness}
Consider strategic behaviors in voting systems like vote manipulation, fraud and bribery, many works have contributed to finding counter-measures for various voting rules. One approach is putting restrictions on voters' preference. Specifically, works of \cite{dummett1961stability,moulin1980strategy,ephrati1991clarke,ephrati1993multi} show that voting with single peak preference and quasilinear preference is truthful and non-manipulatable. Another approach is ensuring computational hardness of finding constructive/deconstructive manipulation strategies (e.g., in \cite{bartholdi1992hard,conitzer2007elections,procaccia2007junta,faliszewski2009hard}). However, for positional voting rules considered in this work, there exist simple greedy algorithms finding strategic votes that manipulate the result in polynomial time \cite{bartholdi1989computational}. There are also some works propose to introduce randomness to the voting process (e.g., sampling voter at random) for mitigating manipulation attacks,  but the usefulness of the voting result is severely harmed \cite{gibbard1977manipulation}.  As a comparison, this work introduces randomness to votes for the purpose of privacy preserving and demonstrates that local differential privacy helps to defend against vote manipulation but makes the voting result more vulnerable to fraudulent votes, when compared to non-private settings (see Section \ref{subsec:soundandprivate} for discussion).

\subsection{Differential Privacy}
\subsubsection{Differential Private Data Analyses}  
As the state-of-the-art data perturbation notion and technique for databases, differential privacy \cite{dwork2011differential} in the centralized setting ensures information-theoretical  privacy. For numerical outputs such as counts and histogram, injecting Laplace random noises \cite{dwork2006calibrating} is the most popular mechanism (e.g., in \cite{li2010optimizing,hay2010boosting, xu2013differentially}). For categorical outputs such as choosing a winning candidate, exponential mechanism \cite{mcsherry2007mechanism} satisfies differential privacy by randomly selecting a category with a probability according to its utility loss (e.g., in \cite{bhaskar2010discovering,steinke2017tight}). For more complex data analyses and mining tasks (e.g, classification learning, clustering), a sequence combination of Laplace and exponential mechanism needs to be used (e.g., in \cite{chen2011publishing,li2012privbasis}).


Despite the functionality of data privacy preservation, differential privacy also has a close relation to stability \cite{dwork2014algorithmic,jain2015drop}, and could avoid false discovery in scientific experiments \cite{hardt2014preventing,dwork2015private,dwork2018differentially}. Since the outputting results are almost indistinguishable when any single individual's data present or not, the exponential mechanism is shown to be sound to data manipulation and data fraud \cite{mcsherry2007mechanism}. However, due to the discrepancy in soundness performance between unbounded and bounded differential privacy \cite{kifer2011no}, local differential private data aggregation is fragile to data fraud attacks (see Section \ref{subsec:soundandprivate} for further discussion). 


\subsubsection{Local Private Data Aggregation}
When defining neighboring datasets as any pairs of values individuals may hold, differential privacy is preserved in the local setting \cite{duchi2013local} (LDP). Because of the solidness of privacy guarantee and flexibility for deployment, LDP has gained massive attention from both industry and academy. Giant internet service providers are collecting user preference (e.g., browser's homepage) and usage records (e.g., typed words) from their users in the local differential privacy manner, such as Google \cite{erlingsson2014rappor,fanti2016building}, Apple \cite{thakurta2017emoji,thakurta2019learning}, and Microsoft \cite{ding2017collecting}. Research works have explored local private data analyses and modeling tasks on various kinds of data, such as distribution estimation on categorical data \cite{bassily2015local,kairouz2016discrete} and set-valued data \cite{qin2016heavy,wang2018privset}, joint distribution estimation and frequent itemset mining on multidimensional data \cite{zhang2018calm,cormode2018marginal}, mean estimation on numerical data \cite{duchi2013local}. There are also theoretical contributions to give lower error bounds of local private data analyses (e.g., in \cite{duchi2013local,kairouz2014extremal,ullman2018tight}). 


Existing works on local differential privacy focus mostly on the usefulness aspect, some of which may also consider computational and communicational efficiency. This work calls for attention to the soundness aspect in local private data analyses, which is severe in real-world systems (e.g., the RAPPOR of Google \cite{erlingsson2014rappor} and iOS/macOS data collection of Apple \cite{tang2017privacy}) where there are malicious and adversarial clients. 

It's worth noting that for some specific voting rules, such as plurality/k-approval voting (see Section \ref{sec:model}), the scored vote can be directly seen as categorical/set-valued data and then processed with existing approaches (e.g., in \cite{kairouz2016discrete,qin2016heavy,wang2018privset}). This work intends to deal with positional voting with the arbitrary design of score vector.  The local private vote aggregation problem can also be cast as the multi-dimensional mean estimation problem, one approach to which is adding Laplace noises to every score (see detail in Section \ref{sec:laplace}), another is first randomly sampling one (data-independent) candidate without knowing the vote and then adding Laplace noises to the candidate's points (e.g., in \cite{nguyen2016collecting,wang2019collecting}). However, the latter approach assumes independence of each vote and can't obtain an unbiased estimator of a whole vote, hence is beyond discussions of this work.
\section{Preliminaries}\label{sec:model}
This section formally introduces definitions of positional voting, local differential privacy and the model of local private vote aggregation. Usefulness and robustness metrics are also defined in this section. We summarize notations throughout this work in Table \ref{tab:notation}. 

\begin{table}
\renewcommand{\arraystretch}{1.17}
\caption{List of notations.}
\label{tab:notation}
\centering
\begin{tabular}{c|l}
\hline
\bfseries Notation & \bfseries Description\\
\hline
$\mathbf{A}$ & The set of candidates/options\\
$d$ & The number of candidates/options\\
$n$ & The number of voters/participants\\
$\mathbf{w}$ & A voting rule's score vector\\
$v^{(i)}$ & A scored vote of voter $i$ that is a permutation of $\mathbf{w}$\\
$\mathbb{D}_v$ & The set of possible permutations of  $\mathbf{w}$\\
$\tilde{v}^{(i)}$ & An estimator (private view) of the scored vote $v^{(i)}$\\
$\mathbb{D}_{\tilde{v}}$ & The set of all possible private views\\
$\theta$ & Average scores of candidates.\\
$\tilde{\theta}$ & Estimator of average scores.\\
$\epsilon$ & The privacy budget\\
\hline
\end{tabular}
\end{table}

\subsection{Positional Voting}
A vote $\pi$ is a linear ordering over all candidates $\mathbf{A}=\{A_1, A_2, ..., A_d\}$, where the relation $\succ$ between two candidates is the preference of a voter. In positional voting, the $j$-th candidate $\pi_j$ in a vote is assigned by a score of $w_j$. For reasonable positional voting rules, the score vector $\mathbf{w}=\{w_1, w_2,..., w_d\}$ is non-increasing, which means $w_j\geq w_{j+1}$. Examples of score vector for popular positional voting rules with $5$ candidates are as follows:
\begin{itemize}
\item{Borda:\ \ \ \ $\{4,\ 3,\ 2,\ 1,\ 0\}$;}
\item{Nauru:\ \ \ \ $\{1/1,\ 1/2,\ 1/3,\ 1/4,\ 1/5\}$; }
\item{Plurality:\ \ \ \ $\{1,\ 0,\ 0,\ 0,\ 0\}$;}
\item{Anti-plurality:\ \ \ \ $\{1,\ 1,\ 1,\ 1,\ 0\}$;}
\item{k-Approval:\ \ \ \ $\{1,\ 1,\ 0,\ 0,\ 0\}\ \ \ (k=2)$.}
\end{itemize}

For the simplicity of reference, we rewrite the voter $i$'s vote $\pi^{(i)}$ as numerical scores for each candidate: $v^{(i)}=[v^{(i)}_1, v^{(i)}_2,..., v^{(i)}_d]$, where $v^{(i)}_j$ is the score of candidate $A_j$.

\subsection{Local Differential Privacy}
The local differential private notion ensures bounded distinguishability in outputs for any two possible inputs, hence blocks adversaries from inferring much information from outputs. Let $\mathcal{D}_\pi$ denote the domain of votes, which represents all possible orderings over candidates $\mathbf{A}$, let $\mathcal{M}$ denote a randomized mechanism, and $\mathcal{D}_\mathcal{M}$ denote the output domain of the mechanism, Definition \ref{def:ldp} formally defines local differential privacy.

\begin{definition}[$\epsilon$-LDP]\label{def:ldp}
A randomized mechanism $\mathcal{M}$ satisfies local $\epsilon$-differential privacy iff for any possible vote pair $\pi,\pi' \in \mathcal{D}_\pi$, and any possible output $t \in \mathcal{D}_\mathcal{M}$,
$$\mathbb{P}[\mathcal{M}(\pi) = t]\leq \exp(\epsilon)\cdot \mathbb{P}[\mathcal{D}_\mathcal{M}(\pi') = t].$$
\end{definition} 
Here the parameter $\epsilon$ is called the privacy budget, which controls the level of privacy preservation. Practical values for $\epsilon$ range between $[0.01,3.0]$.

\subsection{Aggregation Model}
Consider $n$ voters $N=\{1, 2,..., n\}$, the voter $i$ holds a vote $\pi^{(i)}$ (or a scored vote $v^{(i)}$), for the purpose of privacy preservation, the voter sanitizes $\pi^{(i)}$ to get the private view $\tilde{v}^{(i)}$ by running a $\epsilon$-LDP mechanism $\mathcal{M}$ locally and independently. The private view $\tilde{v}^{(i)}$ from a meaningful mechanism is an estimator of the true scored vote $v^{(i)}$, hence the counter in the voting system could estimate the actual average scores $\theta=\frac{1}{n}\sum v^{(i)}$ as:
\begin{equation}\label{eq:estimator}
\tilde{\theta}=\frac{1}{n}\sum \tilde{v}^{(i)}.
\end{equation}

\begin{figure}[!t]
	\centering
	\includegraphics[width=86mm]{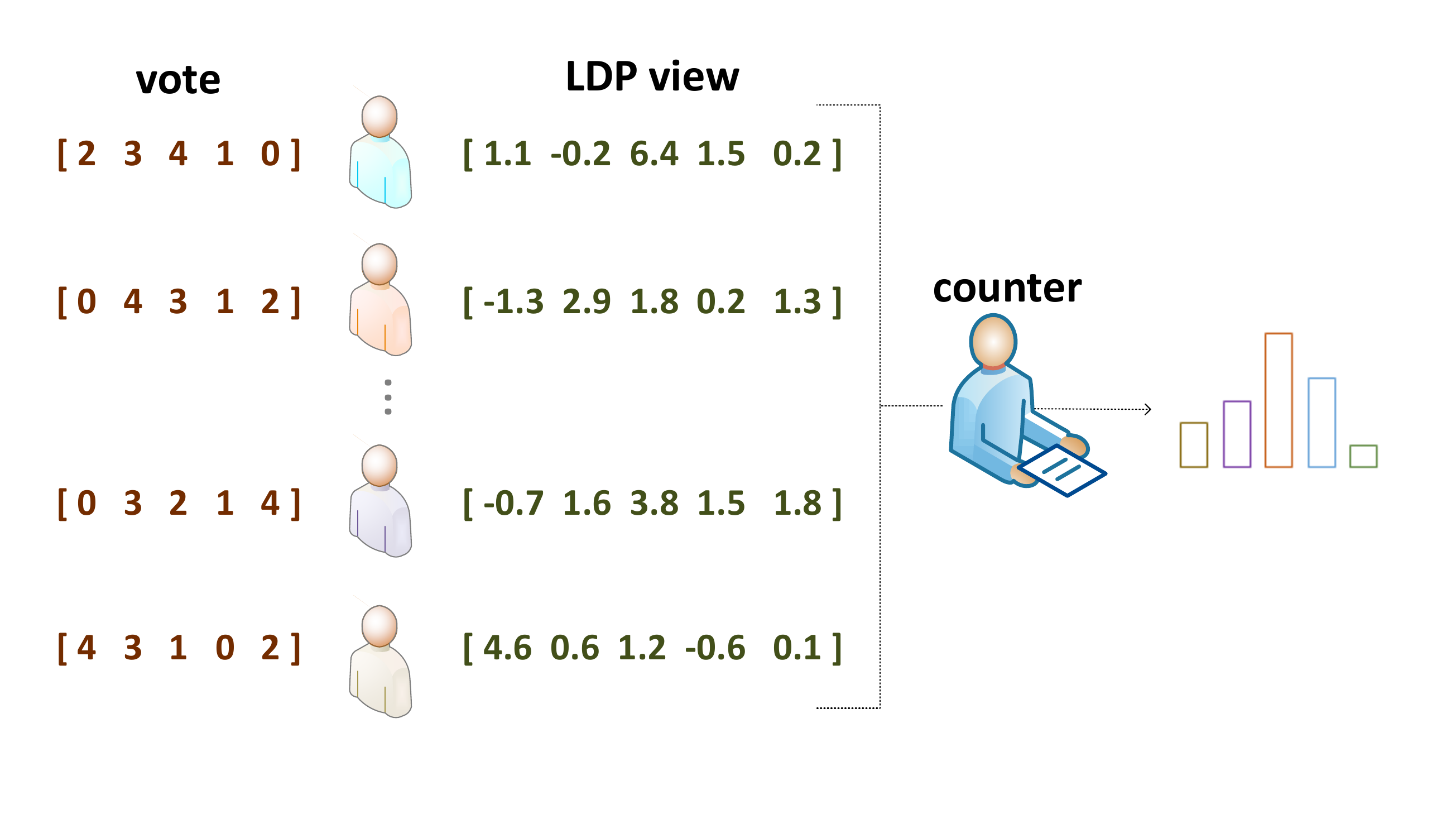}
	\vspace*{-3em}
	\caption{Demonstration of vote aggregation with $\epsilon$-LDP.}
	\label{fig:model}
\end{figure}

Figure \ref{fig:model} demonstrates the above procedures of local private vote aggregation. In adversarial environments, the counter may filter out some potential malicious private views. 

\subsection{Performance Metrics}\label{subsection:metrics}

\ \ \ \ \ \textbf{Usefulness metrics.}\ \ 
Estimators of average scores $\tilde{\theta}$ given by different mechanisms have varied accuracy, here we use three usefulness metrics:
\begin{itemize}
\item{mean squared error: $\text{err}_{\text{MSE}}=\mathbb{E}[|\tilde{\theta}-\theta|_2^2];$}
\item{total variation error: $\text{err}_{\text{TVE}}=\mathbb{E}[|\tilde{\theta}-\theta|_1];$}
\item{maximum absolute error: $\text{err}_{\text{MAE}}=\mathbb{E}[\max_{j\in [1, d]}|\tilde{\theta}_j-\theta_j|].$}
\end{itemize}

Since average scores eventually determine one winning candidate, let $j_{max} = \arg \max_{j\in [1, d]} \theta_{j}$ denote the candidate's index with maximum average score in true average scores $\theta$, and $\tilde{j}_{max}=\arg \max_{j\in [1, d]} \tilde{\theta}_j$ denote the winning candidate's index in the estimated average scores $\tilde{\theta}$, we define following metrics:
\begin{itemize}
\item{accuracy of winner: \ \ $\text{accuracy}_{\text{AOW}}=\mathbb{E}[\tilde{j}_{max} = j_{max}]$;}
\item{loss of winner: \ \ $\text{err}_{\text{LOW}}=\mathbb{E}[\theta_{j_{max}}-\theta_{\tilde{j}_{max}}]$.}
\end{itemize}   

\textbf{Soundness metrics.}\ \   
To measure an adversary's deconstructive power of data amplification attack by contributing one extra vote $v$, we use following metrics:
\begin{itemize}
\item{maximum magnitude: $\text{risk}_{\text{MM}}=\max_{\tilde{v}\in \mathcal{D}_{\tilde{v}}} \frac{|\tilde{v}|_1}{n}$;}
\item{expected magnitude: $\text{risk}_{\text{EM}}=\mathbb{E}[\frac{|\tilde{v}|_1}{n}]$.}
\end{itemize}
These two metrics measure maximum possible/expected absolute difference that one single private view can make to average scores respectively.

For further measuring an adversary's constructive/deconstructive power of view disguise attack by controlling one private view $\tilde{v}$, we define the diameter of the output space $\mathcal{D}_{\tilde{v}}$ that an adversary could choose a value from as follows:
\begin{itemize}
\item{domain diameter: $\text{risk}_{\text{DD}}=\max_{\tilde{v},\tilde{v}'\in \mathcal{D}_{\tilde{v}}} |\tilde{v}-\tilde{v}'|_1$.}
\end{itemize}

\section{Laplace mechanism}\label{sec:laplace}
\subsection{Design}
For numerical values like scored votes, the Laplace mechanism is the most popular approach to (local) differential privacy. The scale of the Laplace noises is calibrated to the privacy budget $\epsilon$ and the sensitivity $\Delta$ of the vote as in Algorithm \ref{alg:laplace}. Lemma \ref{lemma:sensitivity} gives exact bound of the sensitivity $\Delta$ for all positional voting rules, Theorem \ref{theorem:ldplaplace} gives formal $\epsilon$-LDP guarantee.

\begin{lemma}\label{lemma:sensitivity}
For any positional voting rules with non-increasing score vector $\mathbf{w}$, the sensitivity of scored votes is:
$$\Delta=\max_{v,v'\in \mathbb{D}_{v}}|v-v'|_1=\sum_{j\in [1,d]}|w_j-w_{d-j+1}|.$$
\end{lemma}
\begin{proof}
Given that every $v\in \mathbb{D}_v$ is a permutation of the score vector $\mathbf{w}$, we have $\Delta=\max_{v\in \mathbb{D}_{v}}|w-v|_1$. Consider any scored vote $v$, if there exist two indexes $j,j'$ that $j<j'$ and $v_j>v_{j'}$, we denote  $\overrightarrow{v}$ the scored vote that swapped and only swapped $v_j$ and $v_{j'}$ in the scored vote $v$,  then we have:
\begin{equation*}
\begin{aligned}
&&&|w-v|_1-|w-\overrightarrow{v}|_1&\\
&&=&|w_j-v_j|+|w_{j'}-v_{j'}|-|w_j-v_{j'}|-|w_{j'}-v_{j}|&\\
&&=&(|w_j-v_j|-|w_{j'}-v_{j}|)-(|w_j-v_{j'}|-|w_{j'}-v_{j'})|&\\
&&\leq&0,&
\end{aligned}
\end{equation*}
since $f(x)=|w_j-x|-|w_{j'}-x|$ is a non-increasing function for $x\in \mathbb{R}$ when $w_j>w_{j'}$.
By iteratively swapping values that $v_j>v_{j'}$ ($j<j'$) in any scored vote $v$, when $\overrightarrow{v}$ is the reverse of $\mathbf{w}$, we finally have $|\mathbf{w}-v|_1\leq\sum_{j\in [1,d]}|w_j-w_{d-j+1}|$.  
\end{proof}

\begin{theorem}\label{theorem:ldplaplace}
Algorithm \ref{alg:laplace} satisfies $\epsilon$-LDP .
\end{theorem}
\begin{proof}
Because every vote $\pi$ is mapped to one scored vote $v\in \mathbb{D}_v$, to prove $\mathbb{P}[\mathcal{M}(\pi) = t]\leq \exp(\epsilon)\cdot \mathbb{P}[\mathcal{D}_\mathcal{M}(\pi') = t]$ holds for any $\pi$ and $\pi'$, it's enough to show $\text{Pr}[\tilde{v}=t\ |\ v] \leq \exp(\epsilon)\cdot \text{Pr}[\tilde{v}=t\ |\ v']$ holds for any input scored votes $v,v'$ and output $t\in \mathbb{D}_{\tilde{v}}$. By definition of a Laplace random variable $\text{Pr}[Lap(s)=x]=\frac{1}{2s}\exp(-\frac{|x|}{s})$,
we have:
$$\text{Pr}[\tilde{v}=t\ |\ v]=\frac{\epsilon^d}{2^d\Delta^d}\exp(-\frac{\epsilon\cdot |t-v|_1}{\Delta}),$$
hence $\frac{\text{Pr}[\tilde{v}=t\ |\ v]}{\text{Pr}[\tilde{v}=t\ |\ v]}=\exp(\frac{\epsilon\cdot (|t-v'|_1-|t-v|_1)}{\Delta})\leq \exp(\epsilon)$.  
\end{proof}

\begin{algorithm}
    \renewcommand\baselinestretch{1.0}\selectfont
    \caption{Laplace mechanism}
    \label{alg:laplace}
    \begin{algorithmic}[1]
        \Require A scored vote $v\in \mathbb{D}_{v}$, privacy budget $\epsilon$ and the score vector $\mathbf{w}$ of the voting rule.
        \Ensure An unbiased private view $\tilde{v}\in \mathbb{R}^m$ that satisfies $\epsilon$-LDP.
        \State{$\rhd$ Compute sensitivity}
        \State{$\Delta\gets \sum_{j\in [1,d]}|w_j-w_{d-j+1}|$}
        \State{$\rhd$ Randomization by adding Laplace noises}
        \For{$j\gets 1$ \textbf{to} $d$}
            \State{$\tilde{v}_j\gets v_j+Lap(\frac{\Delta}{\epsilon})$}
        \EndFor
        \\\Return{$\tilde{v}=\{\tilde{v}_1,\tilde{v}_2,...,\tilde{v}_d\}$}
    \end{algorithmic}
\end{algorithm}

\subsection{Usefulness Analyses}
The usefulness bound of the average score estimator (refer to Equation \ref{eq:estimator}) given by Laplace mechanism is analyzed in Theorem \ref{theorem:mselaplace}, proof of which is a simple application of Laplace random variables' variance formulation.

\begin{theorem}\label{theorem:mselaplace}
The mean squared error of average score estimator given by the Laplace mechanism in Algorithm \ref{alg:laplace} is:
$$\text{err}_{\text{MSE}}=\frac{2d\cdot(\sum_{j\in [1,d]}|w_j-w_{d-j+1}|)^2 }{n\epsilon^2}.$$ 
\end{theorem}

\subsection{Soundness Analyses}\label{subsec:soundlaplace}
Consider the average score estimator $\tilde{\theta}$ given by the Laplace mechanism, its risks under data amplification attack and view disguise attack are presented in Theorem \ref{theorem:riskslaplace}, proof of which is omitted as being almost trivial.
\begin{theorem}\label{theorem:riskslaplace}
The risks of Laplace mechanism under adversarial attacks are:
\begin{equation*}
\begin{aligned}
&&\text{risk}_{\text{MM}}&=+\infty;&\\ 
&&\text{risk}_{\text{EM}}&=\sum_{j\in[1,d]}\frac{\Delta}{\epsilon}\exp(\frac{-|w_j|\epsilon}{\Delta})+|w_j|;&\\
&&\text{risk}_{\text{DD}}&=+\infty.&
\end{aligned}
\end{equation*}
\end{theorem}

These bounds show that the maximum possible risk of the Laplace mechanism is infinite and the expected risk grow linearly with $\frac{1}{\epsilon}$. Consequently, imposing a stringent level of privacy harms soundness of the voting result. One possible solution to restrict the unlimited maximum possible risk is filtering out private views that are extreme unlikely observed. For example, we may define the allowable output area (with threshold probability $\beta$) as:
$$\tilde{\mathcal{D}}_p=\{\tilde{v}\ |\ \tilde{v}\in \mathcal{R}^m,\ \text{Pr}[\tilde{v}|v] \geq \beta\ \text{for\ some}\  v\in \mathcal{D}_v\}.$$
For Laplace mechanism, it is equivalent to:
$$\{\tilde{v}\ |\ \tilde{v}\in \mathcal{R}^d,\ |\tilde{v}-v|\leq \frac{\Delta(\log(1/\beta)+d\log(\Delta/\epsilon))}{\epsilon}\ \text{for\ some}\  v\in \mathcal{D}_v\}.$$
As a result, even if we can filter out outliers of private views, the diameter or volume of the allowable output area, which determines maximum possible risks, also grows with $\frac{1}{\epsilon}$.

\section{Weighted Sampling Mechanism}\label{sec:sampling}
\subsection{Design}
Another common technique achieving $\epsilon$-LDP for numerical vectors is selecting one (data-dependent) option according to its corresponding value (e.g., for set-valued data \cite{qin2016heavy}, for probabilistic data \cite{kawamoto2018differentially}), and then sanitizing the selected option. This paradigm transforms the numerical $\epsilon$-LDP problem to well-studied categorical one. 

The naive sampling strategy for a scored vote $v$ would be sampling the candidate $j$ with a probability of $\frac{|v_j|}{|v|_1}$ as in \cite{qin2016heavy,kawamoto2018differentially}. Intuitively, the sampling probability should not be related to the absolute magnitude of $v_j$, since adding a constant value to the vector should not change sampling probabilities. Hence we propose a general and flexible weighted sampling strategy with an intercept value $c$ in Algorithm \ref{alg:sampling}. The  state-of-the-art binary randomized response mechanism \cite{duchi2013local,kairouz2016discrete} is used as the base randomizer for later $\epsilon$-LDP protection on the selected categorical candidate. The weighed sampling masses $\mathbf{m}=\{m_1, m_2, ..., m_d\}$ are the vector of sampling probabilities for each rank, we assume $m_j\geq 0.0$ and $\sum_{j\in[1,d]}m_j=1.0$.       

\begin{algorithm}
    \renewcommand\baselinestretch{1.0}\selectfont
    \caption{Weighted sampling mechanism}
    \label{alg:sampling}
    \begin{algorithmic}[1]
        \Require A vote $\pi$, privacy budget $\epsilon$, score vector $\mathbf{w}$ of the voting rule and weighed sampling masses $\mathbf{m}$.
        \Ensure An unbiased private view $\tilde{v}\in \mathbb{R}^d$ that satisfies $\epsilon$-LDP.
        \State{$\rhd$ Select one rank}
        \State{$r\gets UniformRandom(0.0,1.0)$}
        \State{$j^*\gets 0$}
        \While{$r \geq 0.0$}
            \State{$j^*\gets j^*+1$}
            \State{$r\gets r-m_{j^*}$}
        \EndWhile
        \State{$\rhd$ Randomization by binary randomized response}
        \State{$B\gets \{0\}^d$}
        \State{$B_{\pi_{j^*}}\gets 1$}
        \For{$j\gets 1$ \textbf{to} $d$}
            \State{$r\gets UniformRandom(0.0,1.0)$}
            \If{$r< \frac{1}{\sqrt{\exp(\epsilon)}+1}$}
                \State{$\tilde{B}_j\gets 1-B_j$}
            \EndIf
        \EndFor
        \State{$\rhd$ Derive unbiased scores}
        \For{$j\gets 1$ \textbf{to} $d$}
            \State{$\tilde{v}_j\gets \frac{(\sqrt{\exp(\epsilon)}+1)\cdot\tilde{B}_j-1}{\sqrt{\exp(\epsilon)}-1}\cdot\frac{w_{j^*}-c}{m_{j^*}}+c$}
        \EndFor
        
        \\\Return{$\tilde{v}=\{\tilde{v}_1,\tilde{v}_2,...,\tilde{v}_d\}$}
    \end{algorithmic}
\end{algorithm}

Formal $\epsilon$-LDP guarantee of the algorithm is presented in Theorem \ref{theorem:ldpsampling}, unbiasedness of the private view $\tilde{v}$ to the scored vote $v$ is described in Lemma \ref{lemma:unbiasedsampling} (see Appendix \ref{proof:unbiasedsampling} for proof).

\begin{theorem}\label{theorem:ldpsampling}
Algorithm \ref{alg:sampling} satisfies $\epsilon$-LDP .
\end{theorem}
\begin{proof}
Note that the probabilistic rank selection sub-process at line $1$ to $7$  in Algorithm \ref{alg:sampling} uses no information of $\pi$, hence consumes no privacy budget. The only step directly uses information of $\pi$ is at line $10$, which maps the rank $j^*$ to an index  $\pi_{j^*}$ of candidates, thus in order to prove the final private view satisfies $\epsilon$-LDP, it's enough to show that for any $\pi_{j^*},\pi'_{j^*} \in [1,d]$, the corresponding probabilities $Pr[\tilde{v}=t]$ have up to $\exp(\epsilon)$ discrepancy. 

Another observation is that $\tilde{v}$ is derived from the vector $\tilde{B}$ without directly using information about $\pi$, hence we only need to prove resulting vectors $\tilde{B}$ is $\epsilon$-LDP for any $\pi_{j^*},\pi'_{j^*} \in [1,d]$. Follow the proof sketch in \cite{erlingsson2014rappor}, we have:
\small
\begin{equation*}
\begin{aligned}
&&\frac{\text{Pr}[\tilde{B}=T|\pi_{j^*}]}{\text{Pr}[\tilde{B}=T|\pi'_{j^*}]}&=\frac{\sqrt{\exp(\epsilon\cdot[T_{\pi_{j^*}}=1]+\epsilon\cdot[T_{\pi'_{j^*}}=0])}}{\sqrt{\exp(\epsilon\cdot[T_{\pi_{j^*}}=0]+\epsilon\cdot[T_{\pi'_{j^*}}=1])}}&\\
&&&\leq \frac{\sqrt{\exp(\epsilon\cdot 2)}}{\sqrt{\exp(\epsilon\cdot 0)}}\ \ \leq \exp(\epsilon).&
\end{aligned}
\end{equation*} 
\normalsize
\end{proof}
\begin{lemma}\label{lemma:unbiasedsampling}
The private view $\tilde{v}$ given by Algorithm \ref{alg:sampling} is an unbiased estimation of the scored vote $v$.
\end{lemma}

\subsection{Usefulness Analyses}
The accuracy/usefulness of the weighted sampling mechanism depends on the choice of its parameters, including the design of the sampling masses $\mathbf{m}$ and the intercept value $c$. Lemma \ref{lemma:msesampling} formulates the score estimator's error bounds as a function of these parameters (see Appendix \ref{proof:msesampling} for proof), the error can be decomposed into two parts, the first part $\frac{1}{n}\sum_{j\in[1,d]}\frac{(w_{j}-c)^2}{m_{j}}$ is the variance due to weighted sampling, the second part $\frac{1}{n}\frac{d\sqrt{e^\epsilon}}{(\sqrt{e^\epsilon}-1)^2}\sum_{j\in[1,d]}\frac{(w_{j}-c)^2}{m_{j}}$ is the variance due to binary randomized response. Further in Theorem \ref{theorem:minmsesampling}, we establish achievable error bounds via choosing optimal sampling and interception parameters.

Compared to error bounds of Laplace mechanism in Theorem \ref{theorem:mselaplace}, when privacy budget is low (e.g., when $e^\epsilon\approx \epsilon+1$), the factor related to the score vector of a voting rule is improved from $(\sum_{j\in [1,d]}|w_j-w_{d-j+1}|)^2$ to $2(\sum_{j\in[1,d]} |w_{j}-w_{\lceil\frac{d}{2}\rceil}|)^2$. For the simplicity of notation, we define: 
$$\Omega_{\mathbf{w}}=\sum_{j\in[1,d]} |w_{j}-w_{\lceil\frac{d}{2}\rceil}|.$$ 

\begin{figure}[!t]
	\centering
	\includegraphics[width=83mm]{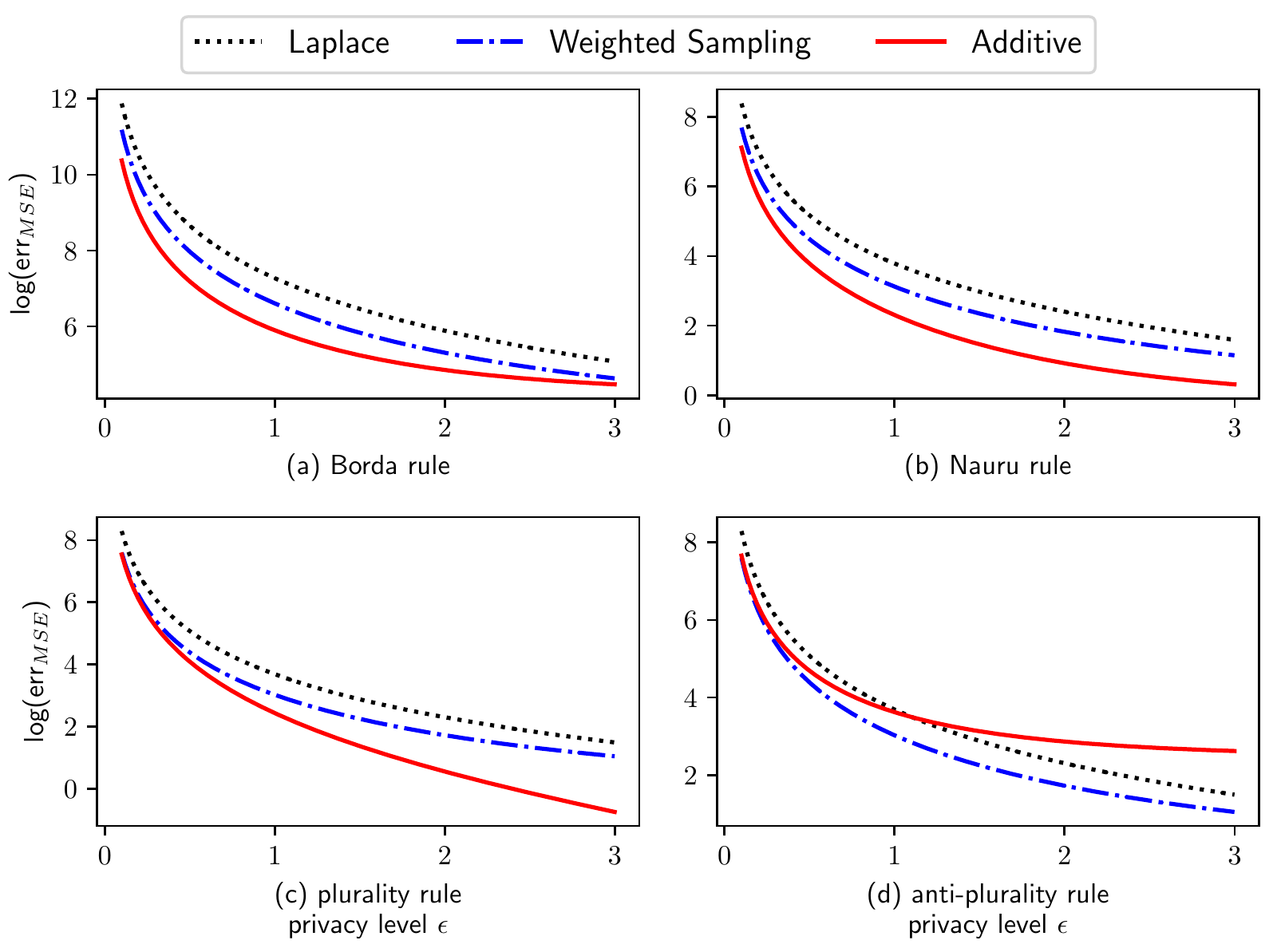}
	\vspace*{-0em}
	\caption{Theoretical mean squared estimation error of Laplace, weighted sampling and additive mechanism with Borda, Nauru, plurality and anti-plurality voting rules over $5$ candidates.}
	\label{fig:themse}
\end{figure}


\begin{lemma}\label{lemma:msesampling}
The mean squared error of average score estimator given by the Laplace mechanism in Algorithm \ref{alg:sampling} with sampling masses $\mathbf{m}$ and intercept value $c$ is:
$$\text{err}_{\text{MSE}}=\frac{1}{n}(1+\frac{d\sqrt{e^\epsilon}}{(\sqrt{e^\epsilon}-1)^2})\sum_{j\in[1,d]}\frac{(w_{j}-c)^2}{m_{j}}.$$
\end{lemma}

\begin{theorem}\label{theorem:minmsesampling}
The mean squared error of average score estimator given by the Laplace mechanism in Algorithm \ref{alg:sampling}  is bounded as follows:
$$\min_{\mathbf{m}\in \mathbb{1}, c\in \mathbb{R}}\text{err}_{\text{MSE}}\leq \frac{1}{n}(1+\frac{d\sqrt{e^\epsilon}}{(\sqrt{e^\epsilon}-1)^2})\cdot(\sum_{j\in[1,d]} |w_{j}-w_{\lceil\frac{d}{2}\rceil}|)^2,$$
corresponding sampling masses $\mathbf{m}^*=[\frac{|w_1-c^*|}{\sum_{|m^*-c^*|}}, \frac{|w_2-c^*|}{\sum_{|m^*-c^*|}}, ..., \frac{|w_d-c^*|}{\sum_{|m^*-c^*|}}]$, where the normalization factor $\sum_{|m^*-c^*|}$ is $\sum_{j\in[1,d]}|w_j-c^*|$, and the intercept value $c^*$ is $\text{median} (\mathbf{w})$ or $w_{\lceil\frac{d}{2}\rceil}$ or $w_{\lfloor\frac{d}{2}\rfloor+1}$.  
\end{theorem}


\begin{proof}
The proof follows two steps, the first step (in Lemma \ref{lemma:msesampling}) writes the mean squared error as a function of sampling masses $\mathbf{m}$ and the intercept value $c$, the second step finds optimal parameters by solving following equivalent problem:
\begin{equation*}
\begin{aligned}
&&\min_{\mathbf{m}, c}&\ \ \ \sum_{j\in[1,d]}\frac{(w_{j}-c)^2}{m_{j}}.&\\
&&s.t.&\ \ \ m_j\geq0.0\ \text{for}\ j\in[1,d]&\\
&&&\ \ \ \sum_{j\in[1,d]}m_j=1.&
\end{aligned}
\end{equation*}

When fixing the variable $c$, the sub-problem $\min_{\mathbf{m}} \sum_{j\in[1,d]}\frac{(w_{j}-c)^2}{m_{j}}$ has closed-form solution of $m_j=\frac{|w_j-c|}{\sum_{j\in[1,d]} |w_{j}-c)|}$. Consequently the optimizing problem becomes:
$$\min_c (\sum_{j\in[1,d]} |w_{j}-c|)^2,$$
which is minimized when $c$ is the median value of the vector $\mathbf{w}$ or $w_{\lceil\frac{d}{2}\rceil}$ or $w_{\lfloor\frac{d}{2}\rfloor+1}$. Substitute the optimal parameters of $\mathbf{m}$ and $c$ into the formula in Lemma \ref{lemma:msesampling}, we have:
$$\text{err}_{\text{MSE}}\leq \frac{1}{n}(1+\frac{d\sqrt{e^\epsilon}}{(\sqrt{e^\epsilon}-1)^2})\cdot(\sum_{j\in[1,d]} |w_{j}-\text{median}(\mathbf{w})|)^2.$$
\end{proof}

\subsection{Soundness Analyses}
Consider a private view $\tilde{v}$ from the weighted sampling mechanism, its risks to the voting result are presented in Lemma \ref{lemma:riskssampling} (see Appendix \ref{proof:riskssampling} for proof). Specifically, with parameters for optimal usefulness (see the former subsection), risk bounds of weighted sampling mechanism are presented in Theorem \ref{theorem:risksoptsampling}. Comparing to the Laplace mechanism, the maximum difference risk and the domain diameter risk are shrunk from $+\infty$ to limited values that grow with $\frac{\sqrt{\exp(\epsilon)}}{\sqrt{\exp(\epsilon)}-1}$. Numerical comparison on expected magnitude risks of the Laplace mechanism and weighted sampling mechanism for Borda voting are presented in Figure \ref{fig:theriskem}.


\begin{theorem}\label{theorem:risksoptsampling}
The manipulation risks of weighted sampling mechanism with  intercept value $c=\lceil \frac{d}{2}\rceil$ and sampling masses $\mathbf{m}=\{\frac{|w_1-c|}{\Omega_{\mathbf{w}}}, \frac{|w_2-c|}{\Omega_{\mathbf{w}}}, ..., \frac{|w_d-c|}{\Omega_{\mathbf{w}}} \}$ are:
\small
\begin{equation*}
\begin{aligned}
&&&\text{risk}_{\text{MM}}=\frac{d}{n}\max[{|\frac{-\sqrt{e^\epsilon}\Omega_{\mathbf{w}}}{\sqrt{e^\epsilon}-1}+c|}, {|\frac{\sqrt{e^\epsilon}\Omega_{\mathbf{w}}}{\sqrt{e^\epsilon}-1}+c|}];&\\
&&&\text{risk}_{\text{EM}}=\sum_{j\in[1,d]}\mathbf{m_j}[w_j>c]\frac{(\sqrt{e^\epsilon}+d-1)t'_{+}+(\sqrt{e^\epsilon}(d-1)+1)f'_{+}}{n\cdot(\sqrt{e^\epsilon}+1)}&\\
&&&\ \ \ \ \ \ \ \ \ \ \ +\sum_{j\in[1,d]}\mathbf{m_j}[w_j<c]\frac{(\sqrt{e^\epsilon}+d-1)t'_{-}+(\sqrt{e^\epsilon}(d-1)+1)f'_{-}}{n\cdot(\sqrt{e^\epsilon}+1)};&\\
&&&\text{risk}_{\text{DD}}=\frac{2\sqrt{e^\epsilon}d}{\sqrt{e^\epsilon}-1}\Omega_{\mathbf{w}}.&
\end{aligned}
\end{equation*}
\normalsize
Where the $t'_+$ and $t'_-$ denote $|\frac{\sqrt{e^\epsilon}}{\sqrt{e^\epsilon}-1}\Omega_{\mathbf{w}}+c|$ and $|-\frac{\sqrt{e^\epsilon}}{\sqrt{e^\epsilon}-1}\Omega_{\mathbf{w}}+c|$ respectively, the $f'_+$ and $f'_-$ denote $|\frac{-1}{\sqrt{e^\epsilon}-1}\Omega_{\mathbf{w}}+c|$ and $|\frac{1}{\sqrt{e^\epsilon}-1}\Omega_{\mathbf{w}}+c|$ respectively.
\end{theorem}
 
\section{Additive Mechanism}\label{sec:additive}
The usefulness/soundness performance of the Laplace mechanism and the weighted sampling mechanism largely depend on the $\Delta_{\mathbf{w}}=\sum_{j\in [1,d]}|w_j-w_{d-j+1}|$ and $\Omega_{\mathbf{w}}=\sum_{j\in[1,d]} |\mathbf{w}_{j}-\mathbf{w}_{\lceil\frac{d}{2}\rceil}|$. When $\mathbf{w}_{\lceil\frac{d}{2}\rceil}$ is relatively close to $\mathbf{w}_d$ (e.g., in $k$-Approval voting and Reciprocal voting), the difference between $\Omega_{\mathbf{w}}$ and $\Delta_{\mathbf{w}}$ will be insignificant. As opposed to the sampling then randomized response paradigm in the weighted sampling mechanism, here we propose an end-to-end approach: the additive mechanism.


\subsection{Design}

Let $\mathcal{C}^k=\{S\ |\ S\subseteq \mathcal{C}\ \text{and}\ |S|=k\}$ denote the set of candidate subsets that have size of $k$, let $\mathbf{w}_{max}^k=\sum_{j\in [1,k]}w_{j}$ denote the maximum total weights of one subset, let $\mathbf{w}_{min}^k=\sum_{j\in [[d-k+1,d]]}w_{j}$ denote the minimum total weights of one subset, the additive mechanism is presented in Definition \ref{def:additive}. 

By its name, the additive mechanism randomly responses with a subset of candidates $S$ with a probability linear to their total scores $\sum_{C_{j'}\in S}v_{j'}$. The mechanism is a novel mutant of the popular exponential mechanism for achieving differential privacy, which usually responses  with a probability proportional to the \textbf{exponential} of candidates' scores. In additive mechanism, the probability is proportional to the \textbf{additive} summary of candidates' scores, which is designed for deriving an unbiased estimation of average scores.

\begin{definition}[Additive Mechanism]\label{def:additive}
In a $\epsilon$-LDP positional voting system, where the candidate is $\mathcal{C}$ and the scored vector is $\mathbf{w}$, take a scored vote $v$ as input, the additive mechanism randomly outputs an $S\in \mathcal{C}^k$ according to following probability design:
$$\text{Pr[S|v]}=\frac{\sum_{C_{j'}\in S}v_{j'}-\mathbf{w}_{min}^k}{\mathbf{w}_{max}^k-\mathbf{w}_{min}^k}\cdot \frac{\exp(\epsilon)-1}{\Phi}+\frac{1}{\Phi},$$
where the normalizer is  $\Phi={d\choose k}\frac{\frac{k}{n}(e^\epsilon-1)\sum_{j\in[1,d]}w_j-e^\epsilon\mathbf{w}_{min}^k+\mathbf{w}_{max}^k}{\mathbf{w}_{max}^k-\mathbf{w}_{min}^k}.$ The estimator of the scored vote $v$ is ($j\in[1,d]$):
$$\tilde{v}_j=a_k\cdot[C_j\in S]-b_k,$$
where $a_k=[\sum_{j'\in[1,d]}w_{j'}(e^\epsilon-1)-\frac{d}{k}e^\epsilon \mathbf{w}_{min}^k+\frac{d}{k}\mathbf{w}_{max}^k]\frac{d-1}{(d-k)(e^\epsilon-1)}$, and $b_k=[\frac{(k-1)(e^\epsilon-1)}{d-1}\sum_{j'\in[1,d]}w_{j'}-e^\epsilon \mathbf{w}_{min}^k+\mathbf{w}_{max}^k]\frac{d-1}{(d-k)(e^\epsilon-1)}$
\end{definition}


After giving formal $\epsilon$-LDP guarantee and unbiasedness guarantee of the additive mechanism in Theorem \ref{theorem:ldpadditive} (see Appendix \ref{proof:ldpadditive} for proof) and Lemma \ref{lemma:unbiasedadditive} respectively, we turn to consider efficient implementation of the additive mechanism, a naive sampling approach would have ${d \choose k}$ computational costs. Actually, selecting a subset of size $k$ from $d$ options in additive mechanism is a stricter case of weighted reservoir sampling \cite{efraimidis2006weighted}, where each option is randomly selected with given marginal weights (probabilities), but no restriction is put on joint probabilities of selected options. We here present a recursive implementation in Algorithm \ref{alg:additive} (see Appendix \ref{appalg:additive}), which decomposes subsets in $\mathcal{C}^k$ into $d-k+1$ groups according to the topmost rank of candidates in a subset, then randomly choose one group and transform to a sub-problem of selecting $k-1$ weighted options. The computational complexity of the recursive algorithm is $O(d\cdot k)$.

\begin{theorem}\label{theorem:ldpadditive}
The additive mechanism satisfies $\epsilon$-LDP.
\end{theorem}

\begin{lemma}\label{lemma:unbiasedadditive}
The private view $\tilde{v}$ given by additive mechanism is an unbiased estimation of the scored vote $v$.
\end{lemma}
\begin{proof}
Consider the probability that an option $C_j$ shows in the output $S$ when input scored vote is $v$:
\begin{equation*}
\begin{aligned}
&&&\ \ \ \sum_{S\in \mathcal{C}^k\ \text{and}\ C_j\in S} \text{Pr}[S|v]&\\
&&&={d-1 \choose k-1}(\frac{v_{j}-\mathbf{w}_{min}^k}{\mathbf{w}_{max}^k-\mathbf{w}_{min}^k}\cdot \frac{\exp(\epsilon)-1}{\Phi}+\frac{1}{\Phi})&\\
&&&\ \ \ +\sum_{j'\in[1,d]\ \text{and}\ j'\neq j}{d-2 \choose k-2}(\frac{v_{j'}-\mathbf{w}_{min}^k}{\mathbf{w}_{max}^k-\mathbf{w}_{min}^k}\cdot \frac{e^\epsilon-1}{\Phi})&\\
&&&=(\frac{{d-1 \choose k-1}v_{j}+{d-2 \choose k-2}[(\sum_{j'\in[1,d]}\mathbf{w}_{j'})-v_j]-\mathbf{w}_{min}^k}{\mathbf{w}_{max}^k-\mathbf{w}_{min}^k}\cdot \frac{e^\epsilon-1}{\Phi}+\frac{{d-1 \choose k-1}}{\Phi}&\\
&&&=\frac{{d-1 \choose k-1}}{{d \choose k}}\frac{\frac{(d-k)(e^\epsilon-1)}{d-1}v_j+\frac{(k-1)(e^\epsilon-1)}{d-1}\sum_{j'\in[1,d]}\mathbf{w}_{j'}-e^\epsilon \mathbf{w}_{min}^k+\mathbf{w}_{max}^k}{\frac{k}{d}\sum_{j'\in[1,d]}\mathbf{w}_{j'}(e^\epsilon-1)-e^\epsilon \mathbf{w}_{min}^k+\mathbf{w}_{max}^k}.&\\
\end{aligned}
\end{equation*}
Consequently, $v_j=a_k\cdot\text{E}[C_j\in S]-b_k$, where $a_k=[\sum_{j'\in[1,d]}\mathbf{w}_{j'}(e^\epsilon-1)-\frac{d}{k}e^\epsilon \mathbf{w}_{min}^k+\frac{d}{k}\mathbf{w}_{max}^k]\frac{d-1}{(d-k)(e^\epsilon-1)}$, and $b_k=[\frac{(k-1)(e^\epsilon-1)}{d-1}\sum_{j'\in[1,d]}\mathbf{w}_{j'}-e^\epsilon \mathbf{w}_{min}^k+\mathbf{w}_{max}^k]\frac{d-1}{(d-k)(e^\epsilon-1)}$. Hence $\text{E}[a_k\cdot[C_j\in S]-b_k]$ is an unbiased estimation of $v_j$.

\end{proof}

\subsection{Usefulness Analyses}
The estimation error bound of the additive mechanism is given in Theorem \ref{theorem:mseadditive}. The formulation of the bound has a dependence on the score vector of a voting rule. In Borda voting, we have $\sum_{j'\in[1,d]}\hat{w}_{j'}=O(d^2)$ when $\epsilon=O(1)$, hence the mean squared error bound is $O(\frac{d^4}{\epsilon^2})$. As a comparison, the mean squared error bounds of the Laplace mechanism and weighted sampling mechanism are  $O(\frac{d^5}{\epsilon^2})$. 

For better illustration, we depict numerical bound of the additive mechanism for Borda rule in Figure \ref{fig:themse},  along with a comparison with the Laplace and weighed sampling mechanism. The numerical results show average $75\%$ error reduction compared to the Laplace mechanism, and average $40\%$ error reduction compared to the weighted sampling mechanism.

\begin{theorem}\label{theorem:mseadditive}
The mean squared error $\mathbb{E}[|\tilde{\theta}-\theta|_2^2]$ of average score estimator given by the additive mechanism is bounded as follows:
$$\frac{(\sum_{j'\in[1,d]}\hat{w}_{j'})^2-\sum_{j'\in[1,d]}\hat{w}_{j'}^2}{n(e^\epsilon-1)^2},$$
where $\hat{w}_{j}=\mathbf{w}_{j}(e^\epsilon-1)- e^\epsilon \mathbf{w}_d+\mathbf{w}_1.$
\end{theorem}
\begin{proof}
The parameter $k=1$ is near to optimal for many voting rules except extremal cases of score vector $\mathbf{w}$ (e.g., of plurality voting). Hence to characterize the usefulness performance of additive mechanism, we only need to analyze the case when $k=1$. Given that $a_1=\sum_{j'\in[1,d]}\mathbf{w}_{j'}-\frac{e^\epsilon d}{e^\epsilon-1}\mathbf{w}_d+\frac{d}{e^\epsilon-1}\mathbf{w}_1$, $b_1=-\frac{e^\epsilon}{e^\epsilon-1}\mathbf{w}_d+\frac{1}{e^\epsilon-1}\mathbf{w}_1$, and $\text{Pr}[C_j\in S|v]=\frac{v_j(e^\epsilon-1)- e^\epsilon \mathbf{w}_d+\mathbf{w}_1}{\sum_{j'\in[1,d]}\mathbf{w}_{j'}(e^\epsilon-1)- e^\epsilon d \mathbf{w}_d+d \mathbf{w}_1}$. The variance of Bernoulli variable $[C_j\in S]$ is $\text{Pr}[C_j\in S|v](1-\text{Pr}[C_j\in S|v])$, then the variance of $\tilde{v}_j=a_1 [C_j\in S]-b_1$ is:
$$(a_1)^2\text{Pr}[C_j\in S|v](1-\text{Pr}[C_j\in S|v]).$$
Consequently, the total variance $\text{E}[|\tilde{v}-v|_2^2]$ is $\frac{(\sum_{j'\in[1,d]}\hat{w}_{j'})^2-\sum_{j'\in[1,d]}\hat{w}_{j'}^2}{(e^\epsilon-1)^2}$.
\end{proof}

\subsection{Soundness Analyses}
The adversarial risks of additive mechanism are presented in Theorem \ref{theorem:risksadditive}. When applying parameter $k=1$ for the Borda rule, we have $a_1=O(\frac{d}{\epsilon})$ and $b_1=O(\frac{d}{\epsilon})$, hence the risk bounds of $\text{risk}_{\text{MM}}$ and $\text{risk}_{\text{EM}}$ are both $O(\frac{d^2}{\epsilon})$, the risk bound of $\text{risk}_{\text{DD}}$ is $O(\frac{d}{\epsilon})$. As a comparison, in the weighted sampling mechanism, the risk bounds of $\text{risk}_{\text{MM}}$ and $\text{risk}_{\text{EM}}$ are $O(\frac{d^3}{\epsilon})$, the risk bound  of $\text{risk}_{\text{DD}}$ is $O(\frac{d^3}{\epsilon})$. Figure \ref{fig:theriskem} presents numerical results of these risks in additive mechanism comparing with the Laplace and weighted sampling mechanisms. In most cases, additive mechanism reduces $70\%$ expected magnitude.

\begin{figure}[!t]
	\centering
	\includegraphics[width=83mm]{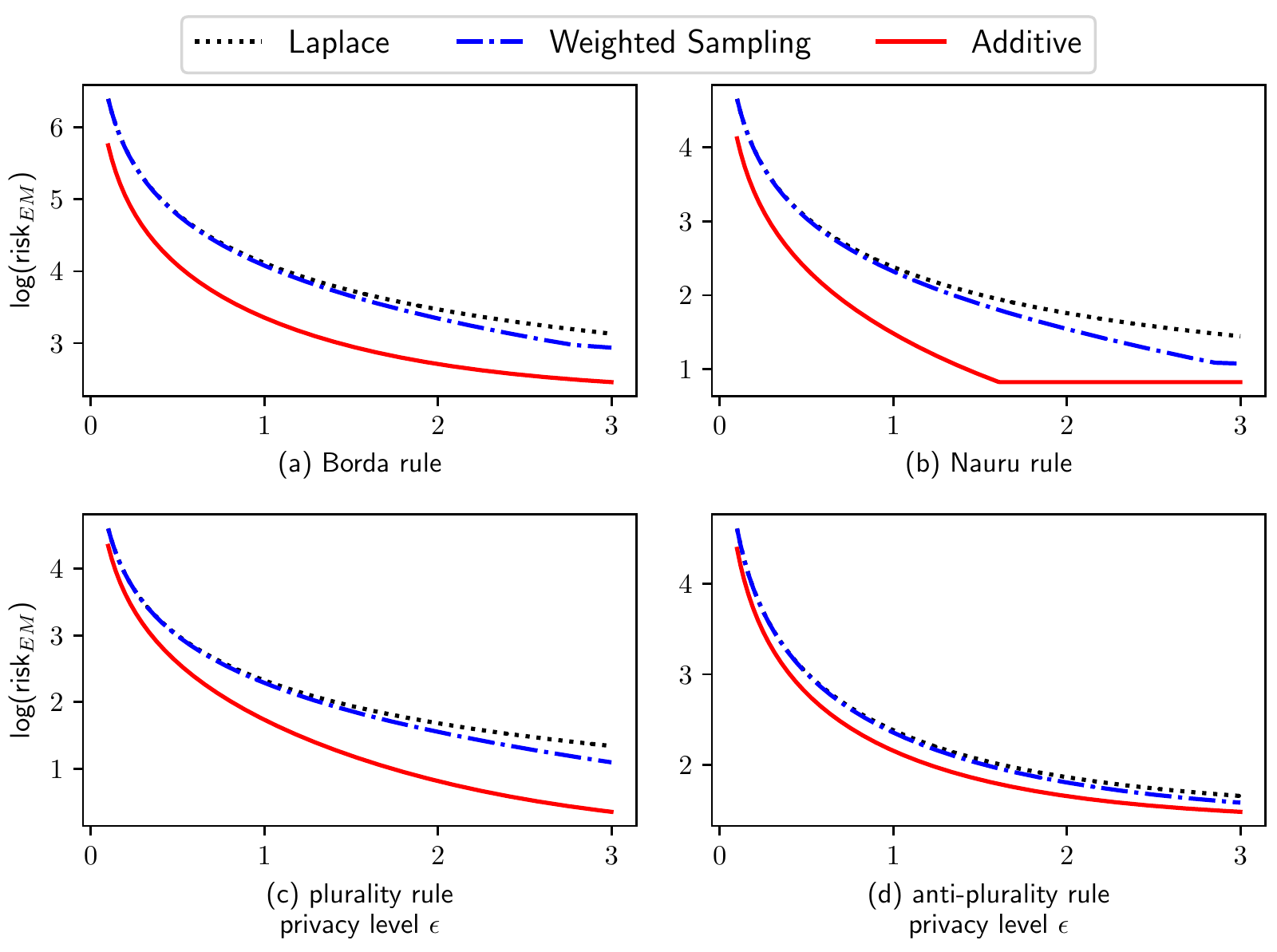}
	\vspace*{-0em}
	\caption{Theoretical expected magnitude risks of Laplace, weighted sampling and additive mechanism with Borda, Nauru, plurality and anti-plurality rules over $5$ candidates.}
	\label{fig:theriskem}
\end{figure}
   
\begin{theorem}\label{theorem:risksadditive}
The manipulation risks of additive mechanism are:
\begin{equation*}
\begin{aligned}
&&&\text{risk}_{\text{MM}} = \frac{k|a_k-b_k|+(d-k)|b_k|}{n};&\\ 
&&&\text{risk}_{\text{EM}}= \frac{k|a_k-b_k|+(d-k)|b_k|}{n};&\\
&&&\text{risk}_{\text{DD}} \leq {2k|a_k|}.&\\
\end{aligned}
\end{equation*}
\end{theorem}
\begin{proof}
For any intermediate result of subset $S\in \mathcal{C}^k$,  the corresponding private view $\tilde{v}=\{\tilde{v}_1,\tilde{v}_2,...,\tilde{v}_d\}$ contains number of $k$ of $a_k-b_k$ and number $d-k$ of $-b_k$, hence both $\text{risk}_{\text{MM}}$ and $\text{risk}_{\text{EM}}$ are $\frac{|a_k-b_k|+(d-1)|b_k|}{n}$. Consider $\text{risk}_{\text{DD}}$, for any paired intermediate results $S,S'\in \mathcal{C}^k$, their corresponding private views differ in at most $2k$ positions, the magnitude of each difference is $|a_k|$, hence the maximum possible total differences are ${2k|a_k|}$.
\end{proof}

\section{Discussion}\label{sec:discussion}
\subsection{Usefulness vs. Indistinguishability}
The privacy budget of $\epsilon$-LDP controls distinguishability in probabilistic outputs, thus limits mutual information between the private view and the scored vote. Vote data, as numerical data with fixed and exclusive ordinal values, could be treated as numerical data (e.g., in Laplace mechanism) or categorical data (e.g., in weighted sampling mechanism), its mean squared error suffers a factor of $\theta(\frac{1}{\epsilon^2})$ when $\epsilon=O(1)$, which is the same as numerical data or categorical data (e.g., lower bounds in \cite{duchi2013local}). Extra factor in mean squared error that is determined by the scored vector of a voting rule and the  concrete design of an $\epsilon$-LDP mechanism. As numerical error bounds in Figure \ref{fig:themse} demonstrated, for most voting rules, the error of weighted sampling mechanism is about $50\%$ of the Laplace mechanism's, while the additive mechanism is about $25\%$.   


\subsection{Soundness vs. Indistinguishability}\label{subsec:soundandprivate}
Since the probabilistic outputs are almost indistinguishable regardless of the value of the manipulated/true vote, a more rigid level of privacy protection in $\epsilon$-LDP has the advantage of limiting an adversary's constructive power,  hence the resulting scores are less led by preferences of the manipulated vote. However, when an adversary could falsely contribute an extra vote to the voting system, a more rigorous level of privacy means a larger magnitude of the private view, thus the expected amount of scores an adversary added to the resulting scores is amplified. As shown in the soundness analyses of Laplace mechanism and our proposed mechanisms, the expected magnitude of private view is linear to $\frac{1}{\epsilon}$, the deconstructive power of every voter (or a possible adversary) gets larger due to the noises injecting for privacy preservation. A desirable property of the additive mechanism is that the maximum possible magnitude of private view equals the expected magnitude, meanwhile the Laplace and weighted sampling mechanisms don't hold. In voting systems for business decision making, a possible counter-measure to the data amplification attack is broadening survey population, increasing number of voters decreases the relative magnitude of one possibly adversarial vote.

Consider the cases of view disguise attack that an adversary directly sends a fraud private view, the ability/power of the adversary is closely related to the domain of the private view, from which the adversary could choose a value to destroy or reform the final result. The private view's domain of Laplace mechanism spreads to  $\mathcal{R}^d$, in the weighted sampling mechanism or the additive mechanism, the domain is reduced to $[-\Theta(\frac{1}{\epsilon}), \Theta(\frac{1}{\epsilon})]^d$. These results suggest that a higher level of privacy preservation  empowers higher ability that an adversary could manipulate the voting result. As counter-measures, imposing a lower level of privacy protection or bringing in more voters could strengthen the soundness of a voting system. Another approach is putting soundness metric into the design of $\epsilon$-LDP mechanism, for example,  risk bounds of the additive mechanism are much better than the Laplace and weighted sampling mechanism's.



Another interesting aspect of soundness is the difference between the centralized differential privacy model and the local differential privacy model. Centralized differential privacy (of the unbounded differential privacy \cite{kifer2011no} notion) assumes all votes have been collected by a trustable database curator, and ensures probability distribution of the final scores of candidates be almost indistinguishable, regardless of whether one single vote is in the database. Consequently, one fraudulent vote can't significantly change the differentially private output in the probabilistic perspective, enforcing a more rigid level of unbounded differential privacy suppresses adversaries' fraud power/income, and is helpful for truthfulness in social welfare maximization \cite{mcsherry2007mechanism}. There is also a notion of differential privacy defined by almost indistinguishable in outputs when one single vote changes value, and is termed bounded differential privacy \cite{kifer2011no}. The $\epsilon$-unbounded differential privacy implies $2\epsilon$-bounded differential privacy, but not vice versa. Local differential privacy is equivalent to the bounded differential privacy defined on a single vote. The $\epsilon$-LDP, unbounded and bounded differential privacy all put limitations on  an adversary's data manipulation power/income, but $\epsilon$-LDP/bounded differential privacy suffers from data fraud.



Inspired by the difference in soundness performance between unbounded differential privacy and bounded (or local) differential privacy, in order to improve soundness against fraudulent votes, the voting counter may force a probability of discard/opt-out probability $\overline{p}$ for every vote, which is similar to data sampling for budget saving in unbounded differential privacy. Consequently, the expected difference a fraudulent vote can make to the final total scores is shrunk by a factor of $\overline{p}$. However, this can't change the expected difference a fraudulent vote can make to final average scores, such as soundness metrics of $risk_{MM}$ and $risk_{EM}$, since the expected number of available voters also shrinks by a factor of $\overline{p}$.

\subsection{Usefulness vs. Soundness}\label{subsec:usefulandsound}
Recall that the weighted sampling mechanism and additive mechanism improve usefulness and soundness metrics simultaneously, compared to the Laplace mechanism. Here we explore interactions between these two performance metrics. Consider the soundness metric of expected magnitude $\text{risk}_{EM}$ and the usefulness metric of mean squared error $\text{err}_{MSE}$, we have:
$$\text{risk}_{EM}\leq\frac{\sqrt{d\cdot n\cdot\text{err}_{MSE}}+\sum_{j\in[1,d]}|\mathbf{w}_j|}{n},$$
which is derived as follows according to convexity of the square root:
\begin{equation*}
\begin{aligned}
&&\frac{\text{E}[|\tilde{v}|_1]}{n}&\leq \frac{\text{E}[|\tilde{v}-v|_1]+|v|_1}{n}\ \leq\ \frac{\text{E}[\sqrt{d\cdot|\tilde{v}-v|_2^2}]+\sum_{j\in[1,d]}|\mathbf{w}_j|}{n}&\\
&&&\leq\ \frac{\sqrt{d\cdot n\cdot\text{E}[|\tilde{\theta}-\theta|_2^2]}+\sum_{j\in[1,d]}|\mathbf{w}_j|}{n}.&\\
\end{aligned}
\end{equation*}
This inequality implies that a mechanism with good usefulness performance usually has good soundness performance too.

In some scenarios that voting administrators pay much attention to soundness performance, they may put hard constraints on $\text{risk}_{DD}$.  A natural question arises, how do soundness constraints affect the usefulness of the voting system? Here we give a negative result that $\epsilon$-LDP mechanism (having unbiased estimator) may not even exist under these constraints. Based on Popoviciu's inequality on variances, we have:
$$\text{risk}_{DD}\geq 2\sqrt{\frac{\text{err}_{MSE}}{n}}.$$
Given that $\text{err}_{MSE}$ goes to infinity as $\epsilon\rightarrow 0$ when the number of voters $n$ is fixed, constraints on $\text{risk}_{DD}$ become unsatisfied. 

  
 

\section{experiments}\label{sec:exp}
We now evaluate the usefulness and soundness performance of weighted sampling and additive mechanism, and compare them with the Laplace mechanism \cite{dwork2006calibrating} and the original sampling-based approach in \cite{qin2016heavy,kawamoto2018differentially}. 

\subsection{Settings}
\ \ \ \ \ \textbf{Datasets.}\ \ In order to thoroughly assess performance of mechanisms in extensive settings, we use synthetic datasets. In each vote aggregation simulation, each candidate $C_j$ is assigned with a uniform random scale $\alpha_j\in[0.0,1.0)$. Each voter's numerical preference $\beta_{(i,j)}$ on candidate $C_j$ is an independent uniform random value $r_{(i,j)}\in [0.0,1.0)$ multiplied by the scale $\alpha_j\in[0.0,1.0)$, then the voter's ranking on candidates is determined by $\beta_{(i,j)}$. In these simulations, the number of candidates $d$ ranges from $4$ to $32$, the number of voters $n$ ranges from $1000$ to $1\ 000\ 000$.

Adversarial votes in the simulation of data amplification attack are uniform-randomly selected from the vote domain $\mathcal{D}_{v}$. The number of adversarial votes $n'$ ranges from $n\cdot 0.1\%$ to $n\cdot 5.0\%$. 

Adversarial private views in the simulation of view disguise attack are generated so that the $2$nd rank candidate $C_{j_2}$ (in the non-adversarial voting result) benefits most. That is, we assume the adversary has the prior knowledge of $1$st and $2$nd ranked candidates, and the adversarial private view $\tilde{v}$ has maximum $\tilde{v}_{j_2}-\tilde{v}_{j_1}$ among the domain of private view. Specifically for the Laplace mechanism that the private view's domain is $[-\infty, +\infty]^d$, we use $95\%$ confidence interval of the Laplace distribution as filtered domain, and assign $\tilde{v}_{j_2}=\log(\frac{1}{1-0.95})\Delta+w_1$, $\tilde{v}_{j_1}=-\log(\frac{1}{1-0.95})\Delta+w_d$. The number of adversarial private views $n''$ in simulations ranges from $n\cdot 0.1\%$ to $n\cdot 5.0\%$. 

Information about simulation parameters is summarized in Table \ref{tab:settings}.  Our experiments focus on the most popular Borda and Nauru voting rules. Every experimental result is the average value of $400$ repeated simulations.

\begin{table}[H]
	\renewcommand{\arraystretch}{1.35}
	\caption{Enumeration of experiment settings, the values in bold format are the default settings.}
	\label{tab:settings}
	\centering
	\begin{tabular}{|c|c|}
		\hline
		\bfseries Parameter & \bfseries Enumerated values\\
		\hline
		voting rule & \textbf{Borda}, Nauru\\
		\hline
		number of candidates $d$ & $4, \mathbf{8}, 16, 32$\\
		\hline
		normal voters $n$ & $1000, \mathbf{10000}, 1000000$\\
		\hline
		adversarial votes $n'$ & $n\cdot 0.1\%,\ \ n\cdot 1.0\%,\ \ n\cdot 5.0\%$\\
		\hline
		adversarial views $n''$ & $n\cdot 0.1\%,\ \ n\cdot 1.0\%,\ \ n\cdot 5.0\%$\\
		\hline
		privacy budget $\epsilon$ &  \small$0.01, 0.1, 0.2, 0.4, 0.8, 1.0, 1.5, 2.0, 3.0$\normalsize\\
		\hline
	\end{tabular}
\end{table}

\textbf{Evaluation Metrics.}\ \ We use usefulness metrics in Section \ref{subsection:metrics} to evaluate the performance of mechanisms in the non-adversarial and adversarial settings. Since the mean squared error and the soundness metrics of mechanisms are theoretically analyzed and numerically compared in former sections, hence their results are omitted.

\begin{figure*}
	\centering
	\includegraphics[width=172mm]{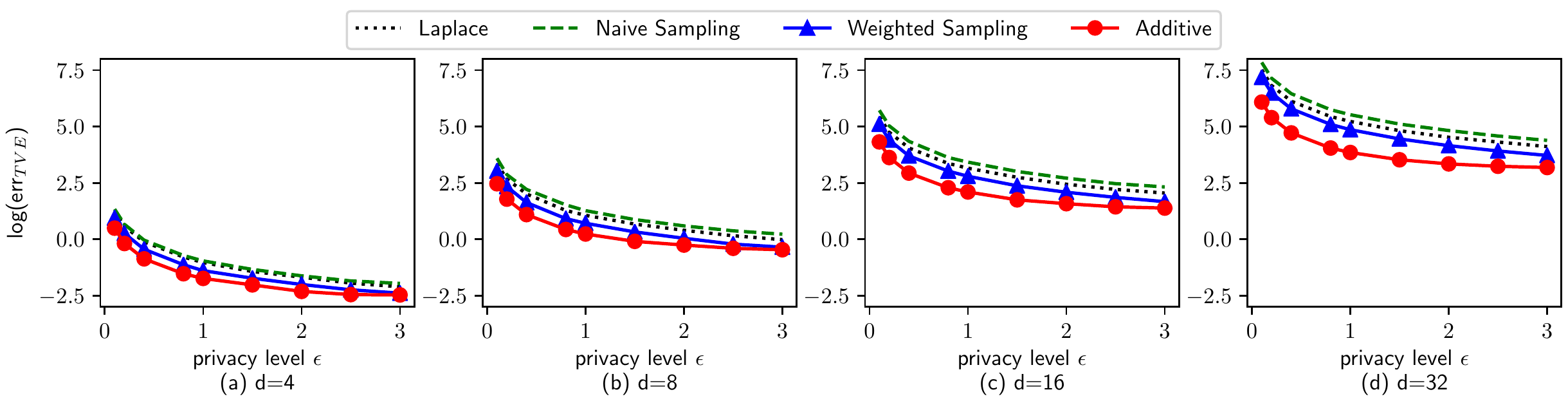}
	\vspace*{-1.5em}
	\caption{Total variation error under Borda rule over $4, 8, 16, 32$ candidates with $10000$ voters.}
	\vspace*{-0.0em}
	\label{fig:tve481632borda}
\end{figure*}

\begin{figure*}
	\centering
	\includegraphics[width=170mm]{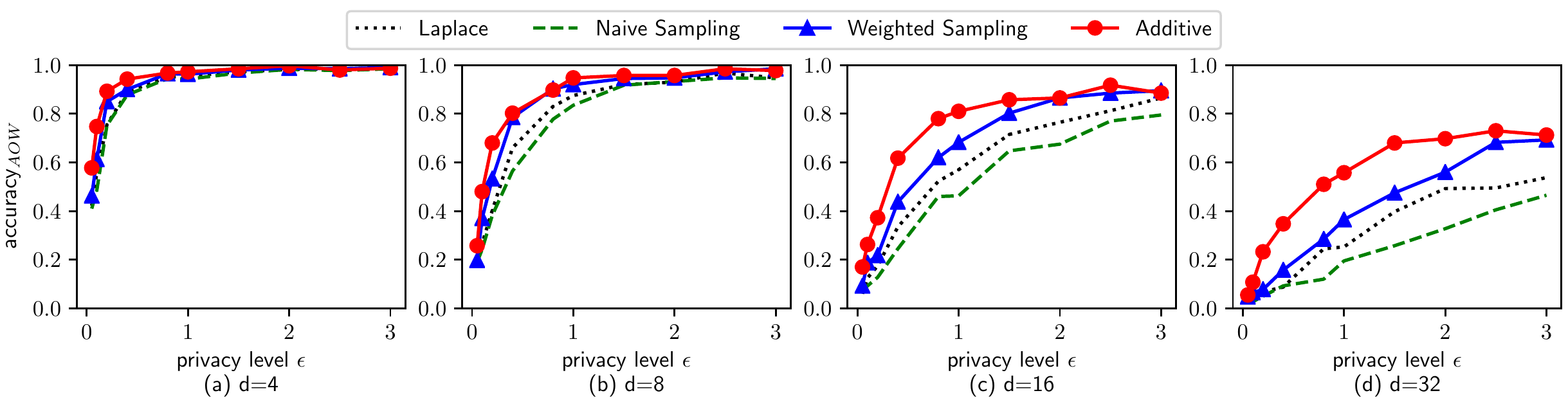}
	\vspace*{-1.5em}
	\caption{Accuracy of winner under Borda rule over $4, 8, 16, 32$ candidates with $10000$ voters.}
	\vspace*{0.5em}
	\label{fig:accuracy481632borda}
\end{figure*}

\begin{figure*}
	\centering
	\includegraphics[width=170mm]{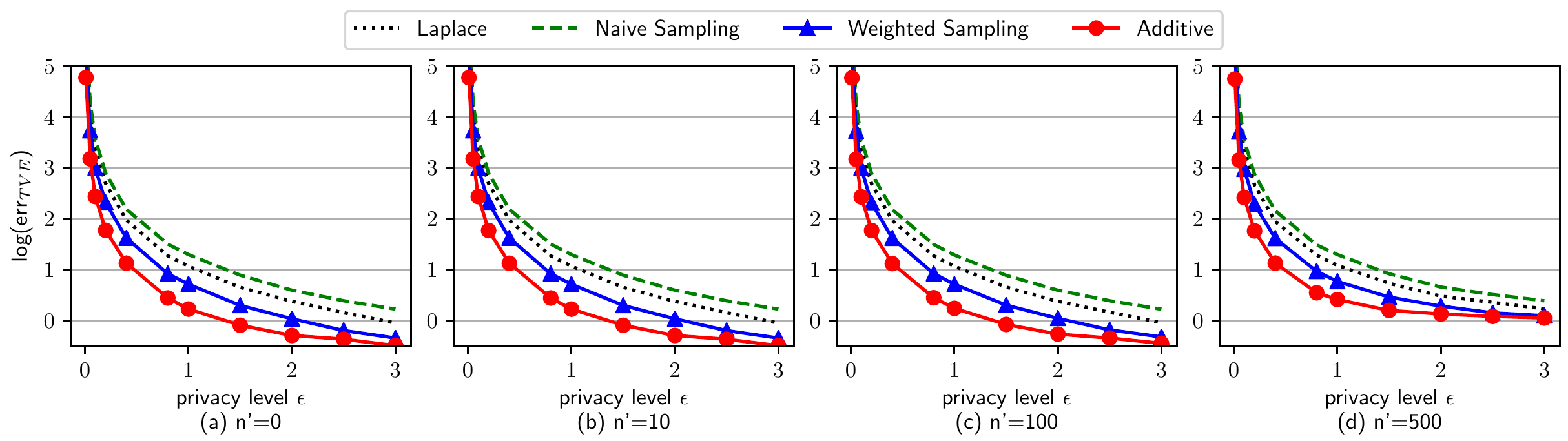}
	\vspace*{-1.5em}
	\caption{Total variation error under Borda rule with $10000$ honest voters and $n'=0, 10, 100, 500$ adversarial votes.}
	\vspace*{-0.0em}
	\label{fig:tvebordana}
\end{figure*}

\begin{figure*}
	\centering
	\includegraphics[width=170mm]{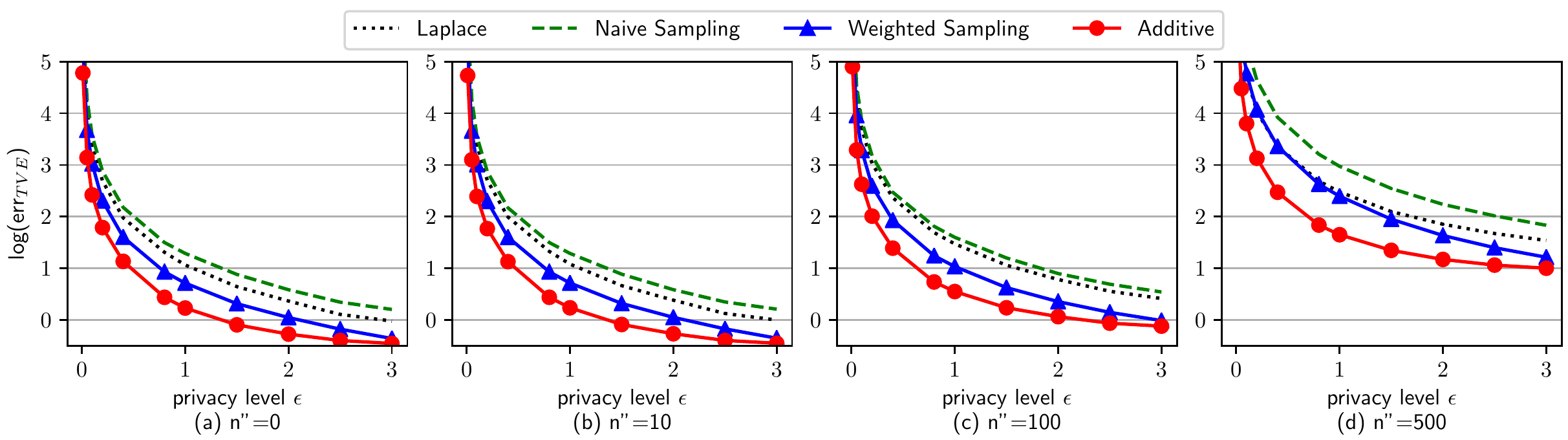}
	\vspace*{-1.5em}
	\caption{Total variation error under Borda rule with $10000$ honest voters and $n'=0, 10, 100, 500$ adversarial private views.}
	\vspace*{-0.0em}
	\label{fig:tvebordand}
\end{figure*}

\subsection{Non-adversarial Results}
\ \ \ \ \ \textbf{Varying number of candidates.}\ \ Simulated with $n=10000$ voters, the total variation error and accuracy of winner results under the Borda rule with varying number of candidates are demonstrated in Figures \ref{fig:tve481632borda} and \ref{fig:accuracy481632borda} respectively (extra results can be found in Appendix \ref{result:candidates}). When compared to the Laplace mechanism, the weighted sampling mechanism averagely reduces $\text{err}_{\text{TVE}}$ by $25\%$, while the additive mechanism averagely reduces $\text{err}_{\text{TVE}}$ by $50\%$. The performance discrepancy between the  weighted sampling and additive mechanism grows with the number of candidates, which confirms our theoretical analyses of mean squared error bounds.


\textbf{Varying number of voters.}\ \ 
Simulated with $d=8$ candidates,  experimental results with $n=1000$ voters are demonstrated in Figures \ref{fig:tve1000} and \ref{fig:accuracy1000}, experimental results with $n=100\ 000$ voters are demonstrated in Figures \ref{fig:tve100000} and \ref{fig:accuracy100000} (additional results can be found in Appendix \ref{result:voters}). Comparing them with results on $n=10000$ voters,  it is observed that increasing the number of voters improves the usefulness performance significantly. The additive mechanism is also practical when there are relatively few voters (e.g., $n=1000$), and  achieve more $80\%$ accuracy when $\epsilon\geq 1.0$.  When the number of voters is $100\ 000$, all mechanisms achieve nearly $100\%$ accuracy of winner even when the privacy budget is low (e.g., $\epsilon<1.0$). 

\begin{figure}
	\centering
	\includegraphics[width=85mm]{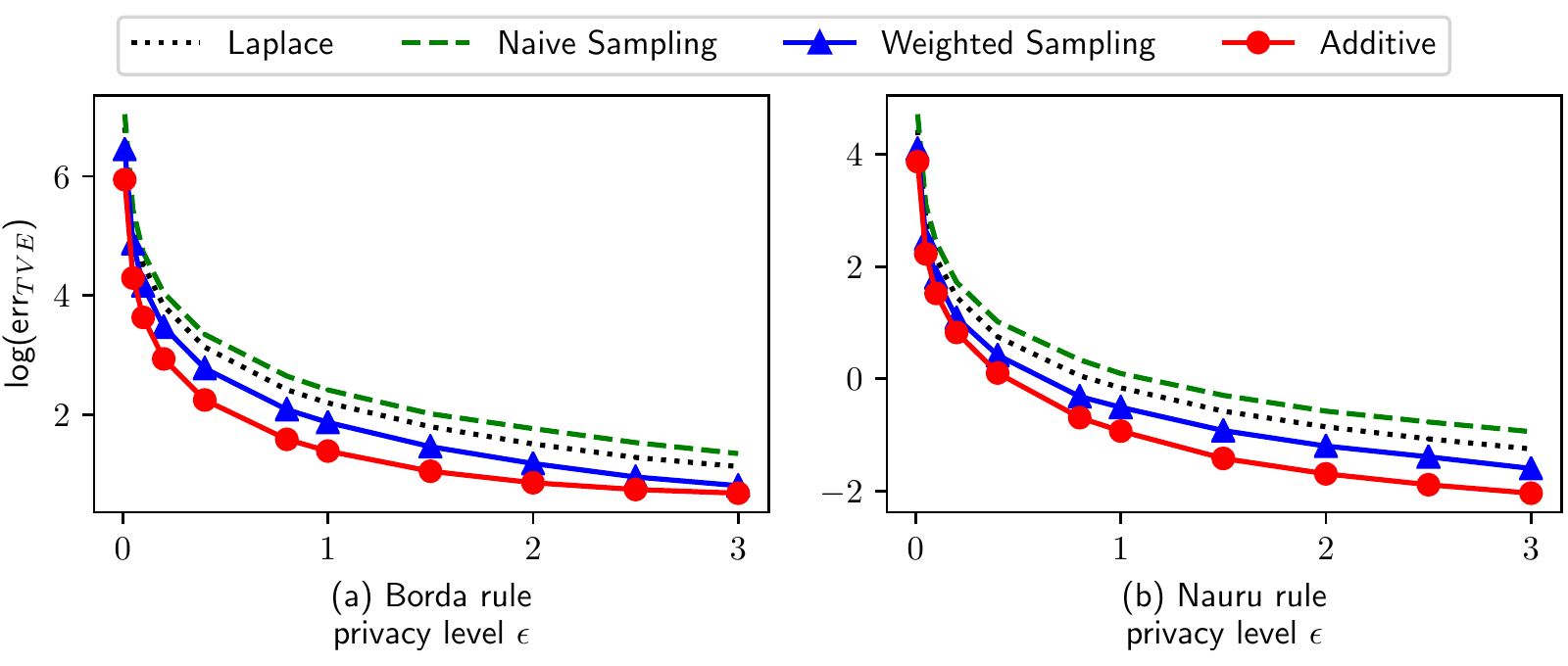}
	\vspace*{-2.7em}
	\caption{Total variation error under Borda and Nauru rules over $8$ candidates with $1000$ voters.}
	\vspace*{-1.0em}
	\label{fig:tve1000}
\end{figure}

\begin{figure}
	\centering
	\includegraphics[width=85mm]{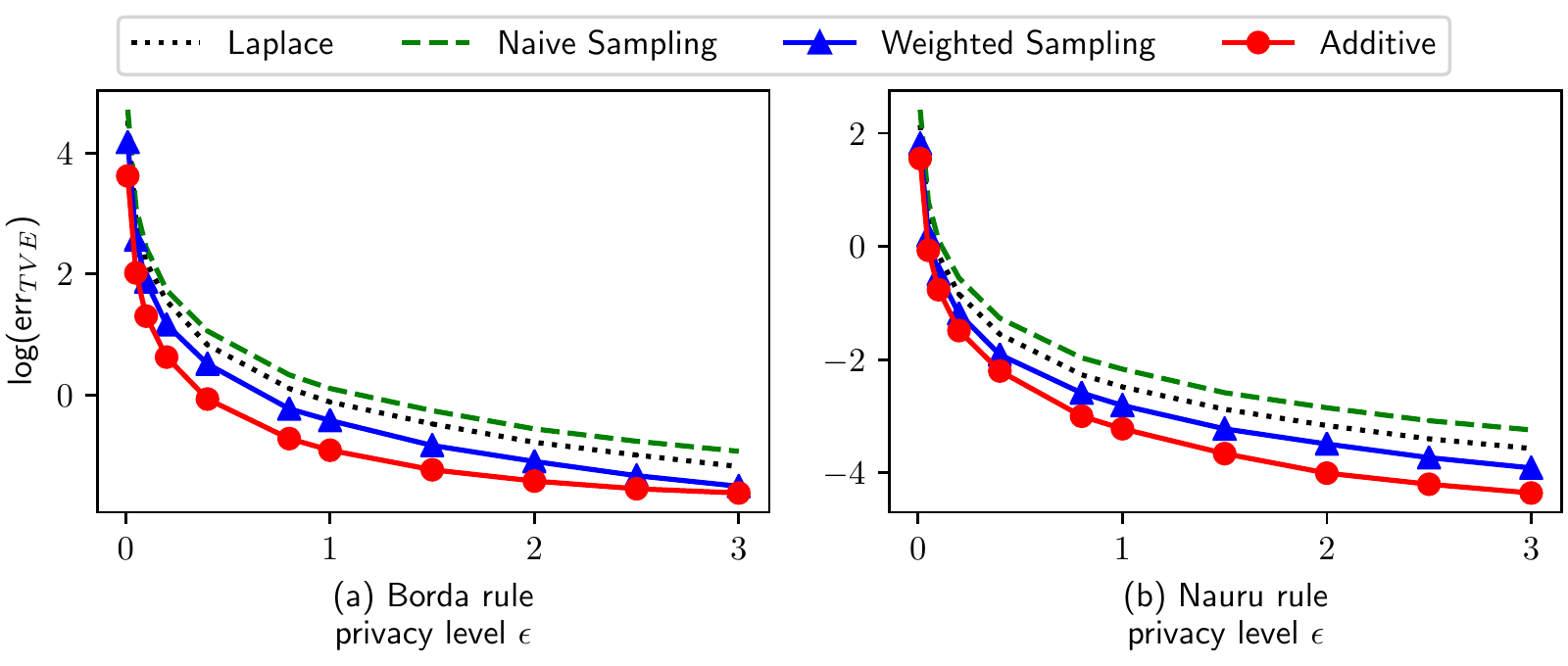}
	\vspace*{-2.7em}
	\caption{Total variation error under Borda and Nauru rules over $8$ candidates with $100000$ voters.}
	\vspace*{-1.1em}
	\label{fig:tve100000}
\end{figure}

\begin{figure}
	\centering
	\includegraphics[width=85mm]{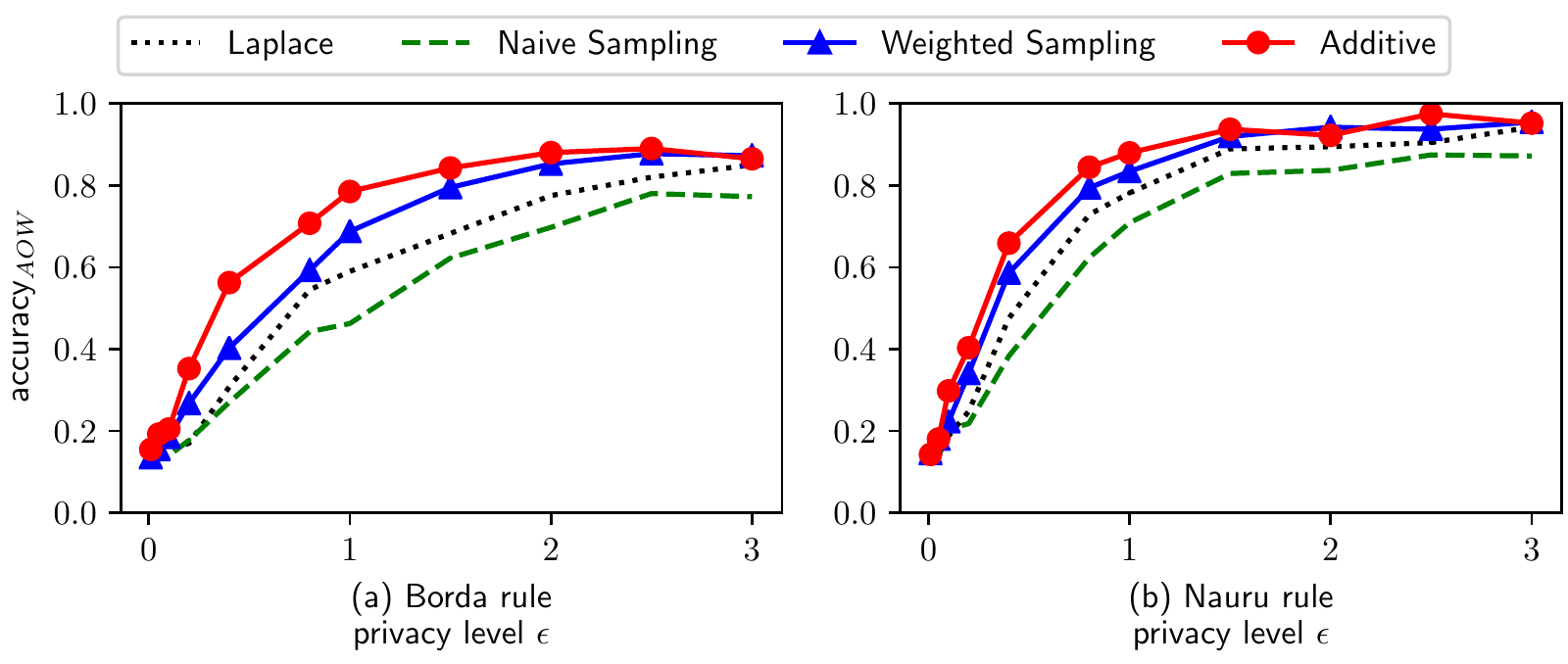}
	\vspace*{-2.7em}
	\caption{Accuracy of winner under Borda and Nauru rules over $8$ candidates with $1000$ voters.}
	\vspace*{-0.7em}
	\label{fig:accuracy1000}
\end{figure}

\begin{figure}
	\centering
	\includegraphics[width=85mm]{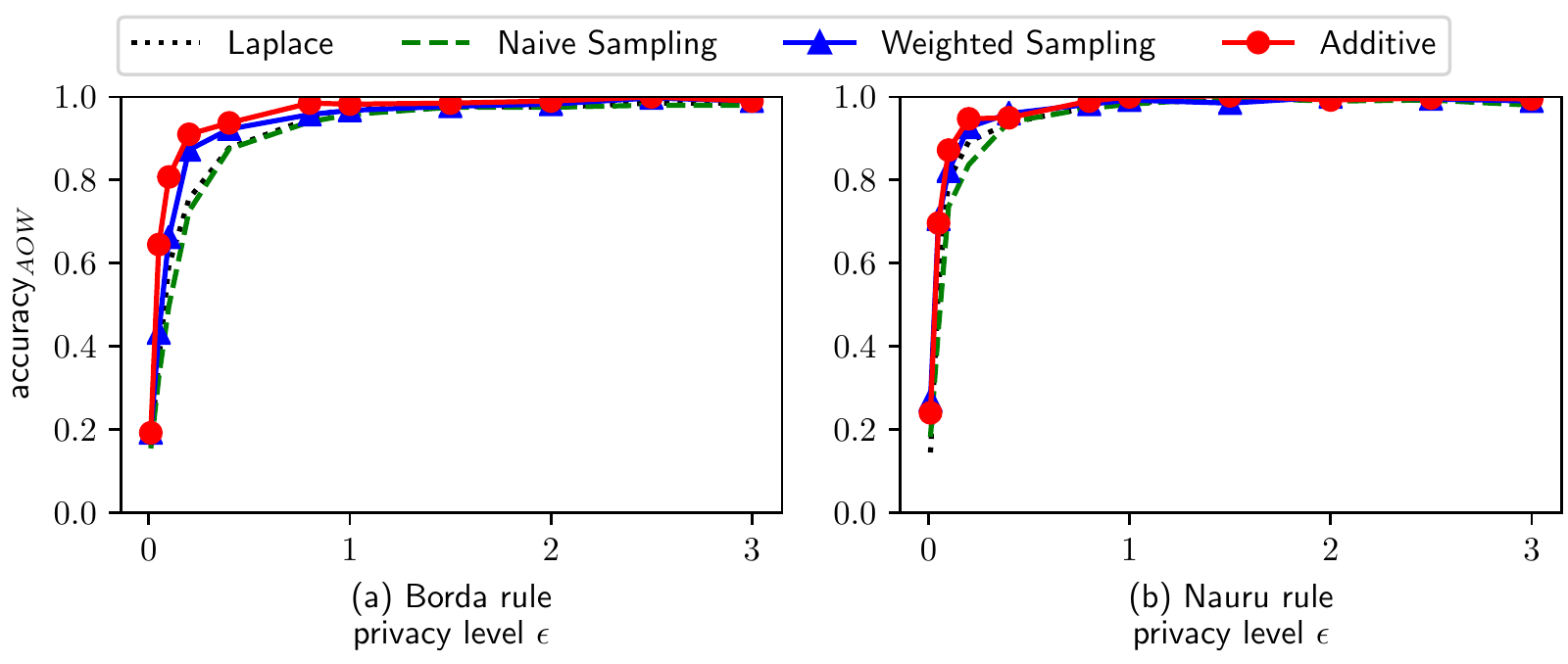}
	\vspace*{-2.7em}
	\caption{Accuracy of winner under Borda and Nauru rules over $8$ candidates with $100000$ voters.}
	\vspace*{-1.0em}
	\label{fig:accuracy100000}
\end{figure}

\subsection{Data Amplification Attack}
Simulated with $d=8$ candidates and $n=10000$ benign voters, the total variation error $\text{err}_{\text{TVE}}$ results under Borda rule with extra $n'=0, 10, 100, 500$ adversarial votes are presented in Figure \ref{fig:tvebordana} (and also in Appendix \ref{result:na}). Results show that less than $1\%$ adversarial votes won't have effective impacts on the voting result, but more than $5\%$ adversarial votes will significantly harm the usefulness of the result. The weighted sampling and the additive mechanisms outperform the Laplace mechanism in all adversarial settings with fraudulent votes.  

\subsection{View Disguise Attack}
Simulated with $d=8$ candidates and $n=10000$ benign voters, the total variation error $\text{err}_{\text{TVE}}$ results under Borda rule with extra $n''=0, 10, 100, 500$ adversarial private views are presented in Figure \ref{fig:tvebordand}  (and also in Appendix \ref{result:nd}). Results show that less than $0.1\%$ adversarial votes won't have effective impacts on the voting result, but more than $1\%$ adversarial votes will significantly decrease the usefulness of the result. The weighted sampling and the additive mechanisms outperform the Laplace mechanism in all adversarial settings with disguised private views. Compared with results under data amplification attacks,  the voting result is more sensitive to view disguise attack.


\section{conclusion}\label{section:conclusion}
Considering adversarial behaviors existing in real-world local private data aggregation systems, this work pays attention to both the usefulness and soundness aspects of privacy preserving mechanisms.  Adversarial behaviors tailed for the local privacy setting are classified into data amplification attack and view disguise attack, which are then quantitatively measured by their manipulation power over the aggregation result. In the context of vote aggregation, two optimized mechanisms: weighted sampling mechanism and additive mechanism, are proposed to improve usefulness and soundness upon the naive Laplace mechanism. Besides theoretical analyses showing a factor of $d$ (or $d^2$) reduction in estimation error bounds and manipulation risk bounds for the Borda voting rule,  their performance improvements are further validated by extensive experiments in both non-adversarial and adversarial scenarios. This work also discusses subtle relations among usefulness, soundness and indistinguishability, and calls for further researches solving dilemmas/conflicts between these fundamental requirements of practical local private data aggregation systems.


%
\bibliographystyle{ACM-Reference-Format}
\clearpage
\bibliography{refs}

%
\clearpage
\section{Appendices}
\subsection{Proof of Lemma \ref{lemma:unbiasedsampling}}\label{proof:unbiasedsampling}
\begin{proof}
We first prove that the numerical vector $\frac{(\sqrt{\exp(\epsilon)}+1)\cdot\tilde{B}-1}{\sqrt{\exp(\epsilon)}-1}$ is an unbiased estimation of the binary vector $B$. Considering an element $\tilde{B}_j$ in $\tilde{B}$, we have:
\begin{equation}\label{eq:unbiasedrandomizer}
\begin{aligned}
&&&\mathbb{E}[\frac{(\sqrt{\exp(\epsilon)}+1)\cdot\tilde{B}_j-1}{\sqrt{\exp(\epsilon)}-1}]&&\\
&&=&\frac{\sqrt{\exp(\epsilon)}}{\sqrt{\exp(\epsilon)}+1}\cdot \frac{(\sqrt{\exp(\epsilon)}+1)\cdot B_j-1}{\sqrt{\exp(\epsilon)}-1}+&\\
&&&+\frac{1}{\sqrt{\exp(\epsilon)}+1}\cdot \frac{-(\sqrt{\exp(\epsilon)}+1)\cdot B_j+\sqrt{\exp(\epsilon)}}{\sqrt{\exp(\epsilon)}-1}&\\
&&=&\frac{(\sqrt{\exp(\epsilon)}-1)\cdot(\sqrt{\exp(\epsilon)}+1)\cdot B_j}{(\sqrt{\exp(\epsilon)}+1)\cdot (\sqrt{\exp(\epsilon)}-1)}&\\
&&=&B_j.&
\end{aligned}
\end{equation}

Next, we want to prove that the vector $B\cdot\frac{\mathbf{w}_{j^*}-c}{\mathbf{m}_j}+c$ is an unbiased estimation of the scored vote $v_j$ defined by $\pi$ as $v_{\pi_j}=\mathbf{w}_j$. Denote a vector $B$ having only the ${j^*}$-th bit being $1$ as $B^{[{j^*}]}$, by the definition of the vector $B$ at line $9$ and $10$, we have:
$$\sum_{{j^*}\in [1,d]} \mathbf{w}_{j^*}\cdot B^{[\pi_{j^*}]}=\sum_{{j^*}\in [1,d]} v_{\pi_{j^*}}\cdot B^{[\pi_{j^*}]}=v.$$
Consequently, we have:
\begin{equation}\label{eq:unbiasedsampling}
\begin{aligned}
&&&\mathbb{E}[B\cdot\frac{\mathbf{w}_{j^*}-c}{\mathbf{m}_{j^*}}+c]&\\
&&=&\sum_{j^*=[1,d]}\mathbf{m}_{j^*}\cdot (B^{[\pi_{j^*}]}\cdot\frac{\mathbf{w}_{j^*}-c}{\mathbf{m}_{j^*}}+c)&\\
&&=&(\sum_{{j^*}\in [1,d]} \mathbf{w}_{j^*}\cdot B^{[\pi_{j^*}]})-(\sum_{{j^*}\in [1,d]}B^{[\pi_{j^*}]}\cdot c)+c&\\
&&=&v-c+c\ \ =v.&\\
\end{aligned}
\end{equation}
Combining unbiasedness results of the randomized response subprocedures in Equation \ref{eq:unbiasedrandomizer} and the weighted sampling sub-procedure in Equation \ref{eq:unbiasedrandomizer}, we conclude that the $\tilde{v}_j$ at line $19$ is an unbiased estimation of $v_j$.     
\end{proof}

\subsection{Proof of Lemma \ref{lemma:msesampling}}\label{proof:msesampling}
\begin{proof}
According to the score estimator's definition in Equation \ref{eq:estimator} and the independence of each $\tilde{v}$ and the unbiasedness of $\tilde{v}$ in Lemma \ref{eq:unbiasedsampling}, we have:
\begin{equation}\label{eq:splitmse}
\begin{aligned}
&&&\mathbb{E}[|\tilde{\theta}-\theta|_2^2]=\mathbb{E}[|\frac{1}{n}\sum_{i\in[1,n]}\tilde{v}^{(i)}-\frac{1}{n}\sum_{i\in[1,n]}v^{(i)}|_2^2]&\\
&&=&\frac{1}{n^2}\sum_{i\in[1,n]}\mathbb{E}[|\tilde{v}^{(i)}-v^{(i)}|_2^2]\ = \frac{1}{n^2}\sum_{j\in[1,d]}\mathbb{E}[(\tilde{v}_{\pi_{j}})^2]-(v_{\pi_{j}})^2.&\\
\end{aligned}
\end{equation}  

Consider the score estimator $\tilde{v}_{\pi_{j}}$ of the $j$-th rank candidate $C_{\pi_{j}}$ in a vote $\pi$, according to sampling strategy, the binary value $B_{\pi_{j}}$ has probability $m_j$ of being $1$ (which happens when $j^*==j$) and has probability $1-m_j$ of being $0$. Further according to the rule of binary randomized response on $B_{\pi_{j}}$, the randomized bit $\hat{B}_{\pi_{j}}$ has probability $\frac{\sqrt{\exp(\epsilon)}}{\sqrt{\exp(\epsilon)}+1}$ of being $1$ when $j^*==j$ and has probability $\frac{1}{\sqrt{\exp(\epsilon)}+1}$ of being $1$ when $j^*\neq j$. Separately considering the random rank $j^*$, we have
\begin{equation*}
\begin{aligned}
&&&\mathbb{E}[(\tilde{v}_{\pi_{j}})^2]&\\
&&=&\sum_{j^*\in[1,d]}m_{j^*}(\frac{[j^*=j]\sqrt{e^\epsilon}}{\sqrt{e^\epsilon}+1}+\frac{[j^*\neq j]}{\sqrt{e^\epsilon}+1}) (\frac{\sqrt{e^\epsilon}}{\sqrt{e^\epsilon}-1}\cdot\frac{\mathbf{w}_{j^*}-c}{\mathbf{m}_{j^*}}+c)^2&\\
&&&+\sum_{j^*\in[1,d]}m_{j^*}(\frac{[j^*=j]}{\sqrt{e^\epsilon}+1}+\frac{[j^*\neq j]\sqrt{e^\epsilon}}{\sqrt{e^\epsilon}+1}) (\frac{-1}{\sqrt{e^\epsilon}-1}\cdot\frac{\mathbf{w}_{j^*}-c}{\mathbf{m}_{j^*}}+c)^2,&
\end{aligned}
\end{equation*}  
summarizing over $j\in[1,d]$, we then have:
\begin{equation}\label{eq:summse}
\begin{aligned}
&&&\sum_{j\in[1,d]}\mathbb{E}[(\tilde{v}_{\pi_{j}})^2]&\\
&&=&\sum_{j^*\in[1,d]}m_{j^*}\frac{\sqrt{e^\epsilon}+d-1}{\sqrt{e^\epsilon}+1} (\frac{\sqrt{e^\epsilon}}{\sqrt{e^\epsilon}-1}\cdot\frac{\mathbf{w}_{j^*}-c}{\mathbf{m}_{j^*}}+c)^2&\\
&&&+\sum_{j^*\in[1,d]}m_{j^*}\frac{1+(d-1)\sqrt{e^\epsilon}}{\sqrt{e^\epsilon}+1} (\frac{-1}{\sqrt{e^\epsilon}-1}\cdot\frac{\mathbf{w}_{j^*}-c}{\mathbf{m}_{j^*}}+c)^2&\\
&&=&\frac{e^\epsilon+d\cdot\sqrt{e^\epsilon}-2\cdot \sqrt{e^\epsilon}+1}{(\sqrt{e^\epsilon}-1)^2}\sum_{j^*\in[1,d]}\frac{(\mathbf{w}_{j^*}-c)^2}{\mathbf{m}_{j^*}}+\sum_{j^*\in[1,d]}\mathbf{w}_{j^*}^2.&\\
\end{aligned}
\end{equation}  

Combining results of Equation \ref{eq:summse} and Equation \ref{eq:splitmse}, we have:
\begin{equation}
\begin{aligned}
&&&\mathbb{E}[|\tilde{\theta}-\theta|_2^2]=\frac{1}{n}(1+\frac{d\sqrt{e^\epsilon}}{(\sqrt{e^\epsilon}-1)^2})\sum_{j\in[1,d]}\frac{(\mathbf{w}_{j}-c)^2}{\mathbf{m}_{j}}.&\\
\end{aligned}
\end{equation}  
\end{proof}

\subsection{Lemma \ref{lemma:riskssampling} on risk bounds of weighted sampling mechanism}\label{proof:riskssampling}

\begin{lemma}\label{lemma:riskssampling}
The manipulation risks of weighted sampling mechanism with sampling masses $\mathbf{m}$ and intercept constant $c$ are:
\small
\begin{equation*}
\begin{aligned}
&&&\text{risk}_{\text{MM}}=\frac{d}{n}\max_{j\in[1,d]}\max[{|\frac{-1}{\sqrt{e^\epsilon}-1}\frac{w_{j}-c}{m_{j}}+c|}, {|\frac{\sqrt{e^\epsilon}}{\sqrt{e^\epsilon}-1}\frac{w_{j}-c}{m_{j}}+c|}];&\\ 
&&&\text{risk}_{\text{EM}}=\frac{\sum_{j\in[1,d]}\mathbf{m_j}[(\sqrt{e^\epsilon}+d-1)t_j+(\sqrt{e^\epsilon}(d-1)+1)f_j]}{n\cdot(\sqrt{e^\epsilon}+1)};&\\
&&&\text{risk}_{\text{DD}}=d(\max_{j^*\in[1,d],\hat{B}_j\in[0,1]} g_{j^*,\hat{B}_j}-\min_{j^*\in[1,d],\hat{B}_j\in[0,1]} g_{j^*,\hat{B}_j}).&\\
\end{aligned}
\end{equation*}
\normalsize
Where $t_j$ denotes $|\frac{\sqrt{e^\epsilon}}{\sqrt{e^\epsilon}-1}\cdot\frac{w_{j}-c}{m_{j}}+c|$, $f_j$ denotes $|\frac{-1}{\sqrt{e^\epsilon}-1}\cdot\frac{w_{j}-c}{m_{j}}+c|$\ and\ $g_{j^*,\hat{B}_j}$ denotes $\frac{[\hat{B}_j=1]\sqrt{e^\epsilon}-[\hat{B}_j=0]}{\sqrt{e^\epsilon}-1}\cdot\frac{w_{j^*}-c}{m_{j^*}}$.
\end{lemma}

\begin{proof}
The proof contains three parts, each part deals with one of the risks in the theorem.

\noindent \textbf{Part $1$} on $\text{risk}_{MM}$:\ \ Recall that the maximum is taken over all possible $j^*\in[1,d]$ and $\hat{B}\in[0,1]^d$. For a given rank $j^*$, apparently the maximum is achieved when $\hat{B}$ is either $[0]^d$ or $[1]^d$, hence we have the result. Additionally, when the intercept value $c$ is no less than $0$, the results can be trimmed to $\frac{d}{n}\max_{j\in[1,d]}{|\frac{\sqrt{e^\epsilon}}{\sqrt{e^\epsilon}-1}\frac{w_{j}-c}{m_{j}}+c|}$. 

\noindent \textbf{Part $2$} on $\text{risk}_{EM}$:\ \ 
Consider the conditional expection $\mathbb{E}[\frac{|\tilde{v}|_1}{n}|j*]$ given rank $j^*$, by the randomization subprocedure of binary randomized response, the randomized vector $\hat{B}=[0,1]^d$ expectedly has $\frac{\sqrt{e^\epsilon}+d-1}{\sqrt{e^\epsilon}+1}$ ones and $\frac{\sqrt{e^\epsilon}(d-1)+1}{\sqrt{e^\epsilon}+1}$ zeros. When $\hat{B}_j$ is $1$, one element $|\tilde{v}_j|$ is $|\frac{\sqrt{e^\epsilon}}{\sqrt{e^\epsilon}-1}\cdot\frac{w_{j^*}-c}{m_{j^*}}+c|$; when $\hat{B}_j$ is $0$, one element $|\tilde{v}_j|$ is $|\frac{-1}{\sqrt{e^\epsilon}-1}\cdot\frac{w_{j}-c}{m_{j}}+c|$. Consequently we have the result.

\noindent \textbf{Part $3$} on $\text{risk}_{DD}$:\ \ Consider one element $j$, we have $\max |\tilde{v}_j-\tilde{v}'_j|_1$ as follows:
\small
$$\max |\frac{[\hat{B}_j=1]\sqrt{e^\epsilon}-[\hat{B}_j=0]}{\sqrt{e^\epsilon}-1}\cdot\frac{w_{j^*}-c}{m_{j^*}}-\frac{[\hat{B}'_j=1]\sqrt{e^\epsilon}-[\hat{B}'_j=0]}{\sqrt{e^\epsilon}-1}\cdot\frac{w_{j^+}-c}{m_{j^+}}|,$$
\normalsize
for any $j^*,j^+\in[1,d],\hat{B}_j,\hat{B}'_j\in[0,1]$. Due to symmetry, the former formula is equivalent to:
$\max_{j^*\in[1,d],\hat{B}_j\in[0,1]} \frac{[\hat{B}_j=1]\sqrt{e^\epsilon}-[\hat{B}_j=0]}{\sqrt{e^\epsilon}-1}\cdot\frac{w_{j^*}-c}{m_{j^*}}-\min_{j^*\in[1,d],\hat{B}_j\in[0,1]} \frac{[\hat{B}_j=1]\sqrt{e^\epsilon}-[\hat{B}_j=0]}{\sqrt{e^\epsilon}-1}\cdot\frac{w_{j^*}-c}{m_{j^*}}$, hence we have the result.
\end{proof}

\subsection{Proof of Theorem \ref{theorem:risksoptsampling}}\label{proof:risksoptsampling}
\begin{proof}
The proof contains three parts, each part deals with one of the risks in the theorem.

\noindent \textbf{Part $1$} on $\text{risk}_{MM}$:\ \ With the given sampling masses $\mathbf{m}$ and interception value $c$, we can derive that $\frac{w_{j}-c}{m_{j}}$ is either $\Omega_{\mathbf{w}}$, $-\Omega_{\mathbf{w}}$ or $0$, then we have the $\text{risk}_{MM}$ as follows:
\small
$$\frac{d}{n}\max[{|\frac{-\sqrt{e^\epsilon}\Omega_{\mathbf{w}}}{\sqrt{e^\epsilon}-1}+c|}, {|\frac{\sqrt{e^\epsilon}\Omega_{\mathbf{w}}}{\sqrt{e^\epsilon}-1}+c|}, {|\frac{\Omega_{\mathbf{w}}}{\sqrt{e^\epsilon}-1}+c|, {|\frac{-\Omega_{\mathbf{w}}}{\sqrt{e^\epsilon}-1}+c|}}].$$
\normalsize 
Since $\sqrt{e^\epsilon}>1$, hence $\max[{|\frac{-\sqrt{e^\epsilon}\Omega_{\mathbf{w}}}{\sqrt{e^\epsilon}-1}+c|}, {|\frac{\sqrt{e^\epsilon}\Omega_{\mathbf{w}}}{\sqrt{e^\epsilon}-1}+c|}]$ is no less than $\max[{|\frac{\Omega_{\mathbf{w}}}{\sqrt{e^\epsilon}-1}+c|, {|\frac{-\Omega_{\mathbf{w}}}{\sqrt{e^\epsilon}-1}+c|}}]$, consequently we have the final results.

\noindent \textbf{Part $2$} on $\text{risk}_{EM}$:\ \ 
When $w_j>c$, the $|\frac{\sqrt{e^\epsilon}}{\sqrt{e^\epsilon}-1}\cdot\frac{w_{j}-c}{m_{j}}+c|$ equals to $t'_+$, and the  $|\frac{-1}{\sqrt{e^\epsilon}-1}\cdot\frac{w_{j}-c}{m_{j}}+c|$ equals to the $f'_+$;  When $w_j<c$, the $|\frac{\sqrt{e^\epsilon}}{\sqrt{e^\epsilon}-1}\cdot\frac{w_{j}-c}{m_{j}}+c|$ equals to $t'_-$, and the $|\frac{-1}{\sqrt{e^\epsilon}-1}\cdot\frac{w_{j}-c}{m_{j}}+c|$ equals to $f'_-$. Consequently we have the final results.

\noindent \textbf{Part $3$} on $\text{risk}_{DD}$:\ \ Recall that $g_{j^*,\hat{B}_j}=\frac{[\hat{B}_j=1]\sqrt{e^\epsilon}-[\hat{B}_j=0]}{\sqrt{e^\epsilon}-1}\cdot\frac{w_{j^*}-c}{m_{j^*}}$, where   $\frac{w_{j}-c}{m_{j}}$ is either $\Omega_{\mathbf{w}}$, $-\Omega_{\mathbf{w}}$ or $0$. When $\Omega_{\mathbf{w}}>0$, the value of $\frac{w_{j}-c}{m_{j}}$ enumerates $[\Omega_{\mathbf{w}},-\Omega_{\mathbf{w}}]$, further because $\sqrt{e^\epsilon}>1$, we have:
$$\max_{j^*\in[1,d],\hat{B}_j\in[0,1]} g_{j^*,\hat{B}_j}=\frac{\sqrt{e^\epsilon}}{\sqrt{e^\epsilon}-1}\Omega_{\mathbf{w}};$$
$$\min_{j^*\in[1,d],\hat{B}_j\in[0,1]} g_{j^*,\hat{B}_j}=-\frac{\sqrt{e^\epsilon}}{\sqrt{e^\epsilon}-1}\Omega_{\mathbf{w}}.$$
Consequently we have the final results.
\end{proof}

\subsection{Proof of Theorem \ref{theorem:ldpadditive}}\label{proof:ldpadditive}
\begin{proof}
Since $\tilde{v}$ is mapped from $S$, to prove the private view $\tilde{v}$ satisfies $\epsilon$-LDP, it's enough to show that the intermediate view $S$ satisfies $\epsilon$-LDP. 

Firstly we need to prove $\text{Pr}[S|v]$ is a valid probability distribution, that is $\text{Pr}[S|v]\geq 0.0$ and $\sum_{S\in\mathcal{C}^k}\text{Pr}[S|v]=1.0$ hold for any input $v\mathcal{D}_{v}$. Since $\sum_{C_{j'}\in S}v_{j'}\geq \mathbf{w}_{min}^k$, we have $\Phi>0$ and hence $\text{Pr}[S|v]\geq 0.0$. Now consider $\sum_{S\in\mathcal{C}^k}\text{Pr}[S|v]$, we have:
\begin{equation*}
\begin{aligned}
&&&\ \ \frac{{d \choose k}}{\Phi}+\sum_{S\in\mathcal{C}^k}\sum_{C_{j'}\in S}\frac{v_{j'}-\mathbf{w}_{min}^k}{\mathbf{w}_{max}^k-\mathbf{w}_{min}^k}\cdot \frac{\exp(\epsilon)-1}{\Phi}&\\
&&&=\frac{{d \choose k}}{\Phi}+{d-1 \choose k-1}\sum_{C_{j'}\in\mathcal{C}}\frac{v_{j'}-\mathbf{w}_{min}^k}{\mathbf{w}_{max}^k-\mathbf{w}_{min}^k}\cdot \frac{\exp(\epsilon)-1}{\Phi}&\\
&&&=\frac{{d \choose k}}{\Phi}+{d-1 \choose k-1}\frac{\sum_{j\in[1,d]}w_{j}-\mathbf{w}_{min}^k}{\mathbf{w}_{max}^k-\mathbf{w}_{min}^k}\cdot \frac{\exp(\epsilon)-1}{\Phi}&\\
&&&=\frac{1}{{\Phi}}{d\choose k}\cdot(\frac{\frac{k}{d}\sum_{j\in[1,d]}\mathbf{w}_j-\mathbf{w}_{min}^k}{\mathbf{w}_{max}^k-\mathbf{w}_{min}^k}\cdot (e^\epsilon-1)+1)\ =\ 1.&\\
\end{aligned}
\end{equation*}

Secondly for any paired inputs $v,v'\in \mathcal{D}_{v}$ and any output values $S\in \mathcal{C}^k$, we have:
\begin{equation*}
\begin{aligned}
&&&\ \ \frac{\text{Pr}[S|v]}{\text{Pr}[S|v']}\leq \frac{\max_{S\in \mathcal{C}^k} \text{Pr}[S|v]}{\min_{S\in \mathcal{C}^k} \text{Pr}[S|v']}&\\
&&&\leq (\frac{\mathbf{w}_{max}^k-\mathbf{w}_{min}^k}{\mathbf{w}_{max}^k-\mathbf{w}_{min}^k}\cdot \frac{e^\epsilon-1}{\Phi}+\frac{1}{\Phi})/(\frac{\mathbf{w}_{min}^k-\mathbf{w}_{min}^k}{\mathbf{w}_{max}^k-\mathbf{w}_{min}^k}\cdot \frac{e^\epsilon-1}{\Phi}+\frac{1}{\Phi})&\\
&&&\leq \exp(\epsilon).
\end{aligned}
\end{equation*}
\end{proof}

\subsection{Algorithm of additive mechanism}\label{appalg:additive}
The algorithm implementation \ref{alg:additive} for additive mechanism (in Definition \ref{def:additive}) contains a core procedure: $additive\_select,$ which recursively select a top ranking position  $j^*$ from remaining positions $[1,d]$ (see Algorithm \ref{alg:additiveselect}). The relative weights $z_j$ in the algorithm are proportional to the probability that the candidate $C_j$ will show in the private view $S$.
\begin{algorithm}
    \renewcommand\baselinestretch{1.0}\selectfont
    \caption{Additive mechanism}
    \label{alg:additive}
    \begin{algorithmic}[1]
        \Require A vote $\pi$, privacy budget $\epsilon$,voting rule's score vector $\mathbf{w}$ and parameter $k$.
        \Ensure An unbiased private view $\tilde{v}\in \mathbb{R}^d$ that satisfies $\epsilon$-LDP.
        \State{$\rhd$ Select $k$ ranking positions}
		\For{$j\in[1,d]$}
			\State{$\rhd$ Compute weights of presence for a ranking position}
			\State{$z_j \gets \frac{w_j-w^k_{min}/k}{w^k_{max}-w^k_{min}}\cdot (e^\epsilon-1)+\frac{1}{k}$}
		\EndFor        
        \State{$T \gets additive\_select(d, k, \mathbf{z})$}
        \State{$\rhd$ Deriving unbiased estimator}
        \State{$S \gets \{\pi_j\ |\  j\in T\}$}
        \For{$j \in [1,d]$}
        	\State{$\tilde{v}_j \gets [C_j\in S]\cdot a_k-b_k$}
        \EndFor
        \\\Return{$\tilde{v}=\{\tilde{v}_1,\tilde{v}_2,...,\tilde{v}_d\}$}
    \end{algorithmic}
\end{algorithm}

\begin{algorithm}[H]
    \renewcommand\baselinestretch{1.0}\selectfont
    \caption{additive\_select$(d, k, \mathbf{w})$}
    \label{alg:additiveselect}
    \begin{algorithmic}[1]
    	\Require The number of positions $d$, parameter $k$, and positions' weights $\mathbf{z}$.
        \Ensure $k$ ranking positions $T\subseteq [1,d]$.
        \State{$\rhd$ Compute probabilities of $p_j=\text{Pr}[min(T)=j]$}
        \For{$j \in [1,d-k+1]$}
        	\State{$p_j \gets {d-j \choose k-1}\cdot(z_j+(\sum_{j'\in[j+1,d]}z_{j'}-z_j)\frac{k-1}{d-j})$}
        \EndFor
        \State{$\rhd$ Select a minimum rank $j^*$}
        \State{$j^*\gets 0$}
        \While{$r \geq 0.0$}
            \State{$j^*\gets j^*+1$}
            \State{$r\gets r-\frac{p_{j^*}}{\sum_{j\in[1,d]}p_{j^*}}$}
        \EndWhile
        \For{$j\in [j^*+1,d]$}
        	\State{$z'_{j-j^*}=w_j+\frac{z_{j^*}}{k-1}$}
        \EndFor
         \State{$\rhd$ Recursively select $k-1$ ranking positions}
        \State{$T' = additive\_select(d-j^*, k-1, \mathbf{z}')$}
        \\\Return{$T=\{j^*\}\cup \{j+j^*\ |\ j\in T'\}$}
        \end{algorithmic}
\end{algorithm}

\subsection{Additional experimental results of varying number of voters}\label{result:voters}
Results of maximum absolute error and loss of winner error under the Borda rule with $n=1000$ votes are demonstrated in Figures \ref{fig:mae1000} and \ref{fig:low1000} respectively. The maximum absolute error and loss of winner error results under the Borda rule with $n=100000$ votes are demonstrated in Figures \ref{fig:mae100000} and \ref{fig:low100000} respectively.

\begin{figure}[h]
	\centering
	\includegraphics[width=85mm]{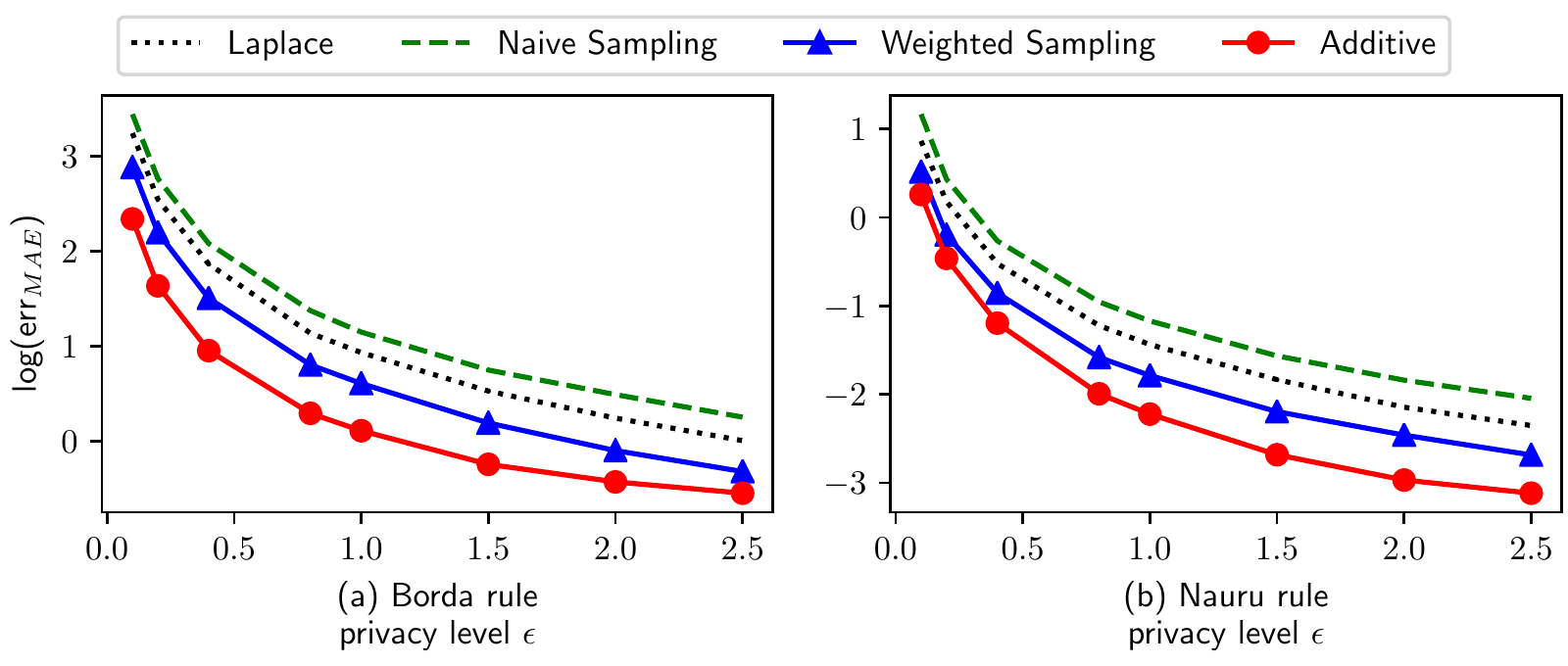}
	\vspace*{-2.0em}
	\caption{Maximum absolute error under Borda and Nauru rules over $8$ candidates with $1000$ voters.}
	\vspace*{-0.5em}
	\label{fig:mae1000}
\end{figure}

\begin{figure}[h]
	\centering
	\includegraphics[width=85mm]{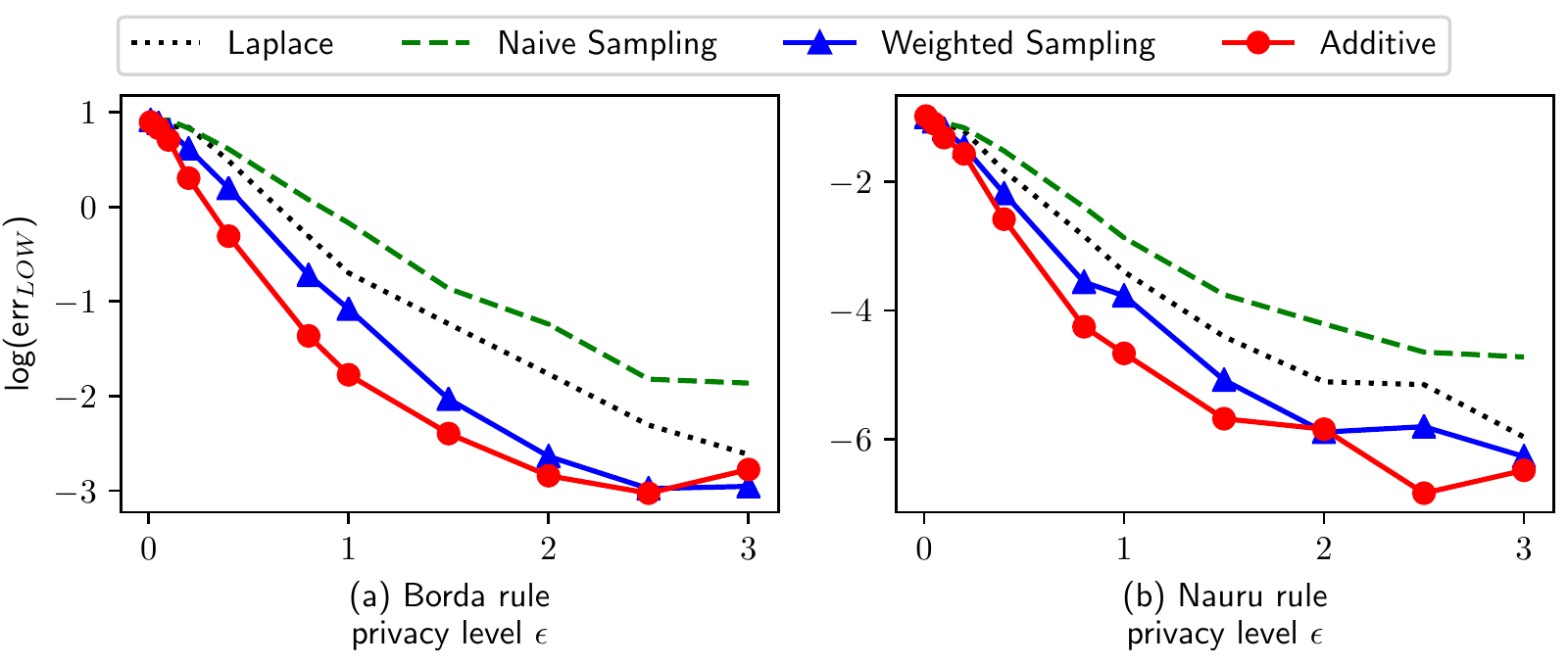}
	\vspace*{-2.0em}
	\caption{Loss of winner error under Borda and Nauru rules over $8$ candidates with $1000$ voters.}
	\vspace*{-0.5em}
	\label{fig:low1000}
\end{figure}

\begin{figure}[h]
	\centering
	\includegraphics[width=85mm]{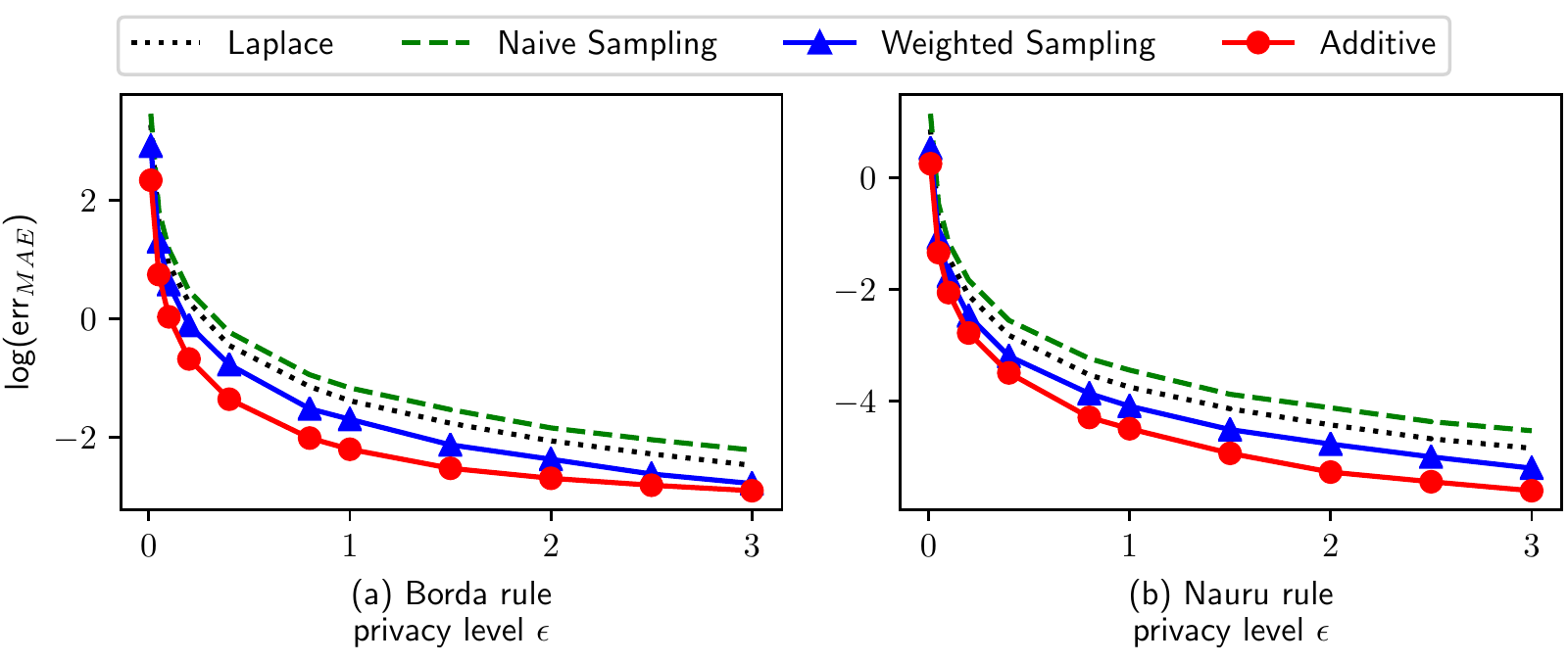}
	\vspace*{-2.0em}
	\caption{Maximum absolute error under Borda and Nauru rules over $8$ candidates with $100000$ voters.}
	\vspace*{-0.5em}
	\label{fig:mae100000}
\end{figure}

\begin{figure}[h]
	\centering
	\includegraphics[width=85mm]{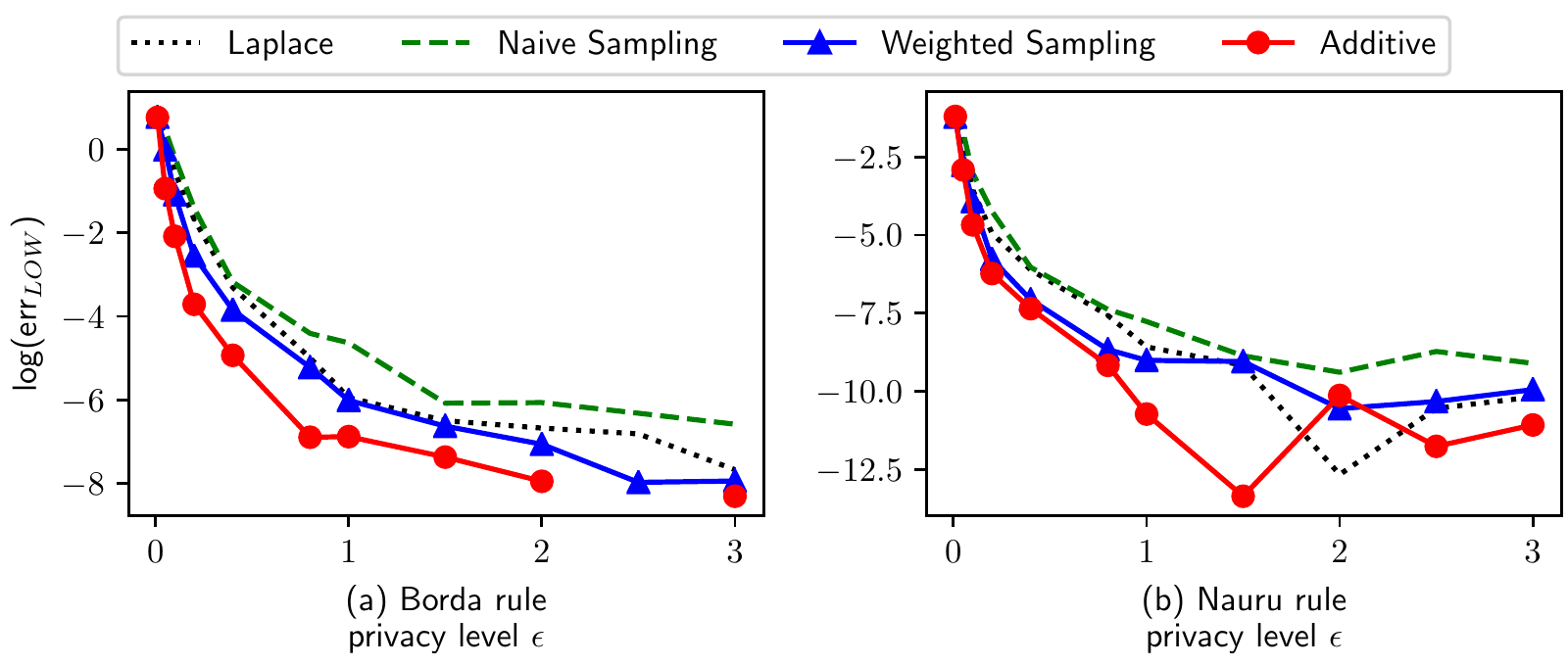}
	\vspace*{-2.0em}
	\caption{Loss of winner error under Borda and Nauru rules over $8$ candidates with $100000$ voters.}
	\vspace*{-0.5em}
	\label{fig:low100000}
\end{figure}


\subsection{Additional experimental results of varying number of candidates}\label{result:candidates}
Results of maximum absolute error and loss of winner error under the Borda rule with varying number of candidates are demonstrated in Figures \ref{fig:mae481632borda} and \ref{fig:low481632borda} respectively. The total variation error, maximum absolute error, accuracy of winner and loss of winner error results under the Nauru rule with varying number of candidates are demonstrated in Figures \ref{fig:tve481632nauru}, \ref{fig:mae481632nauru}, \ref{fig:accuracy481632nauru} and \ref{fig:low481632nauru} respectively.

\begin{figure*}
	\centering
	\includegraphics[width=175mm]{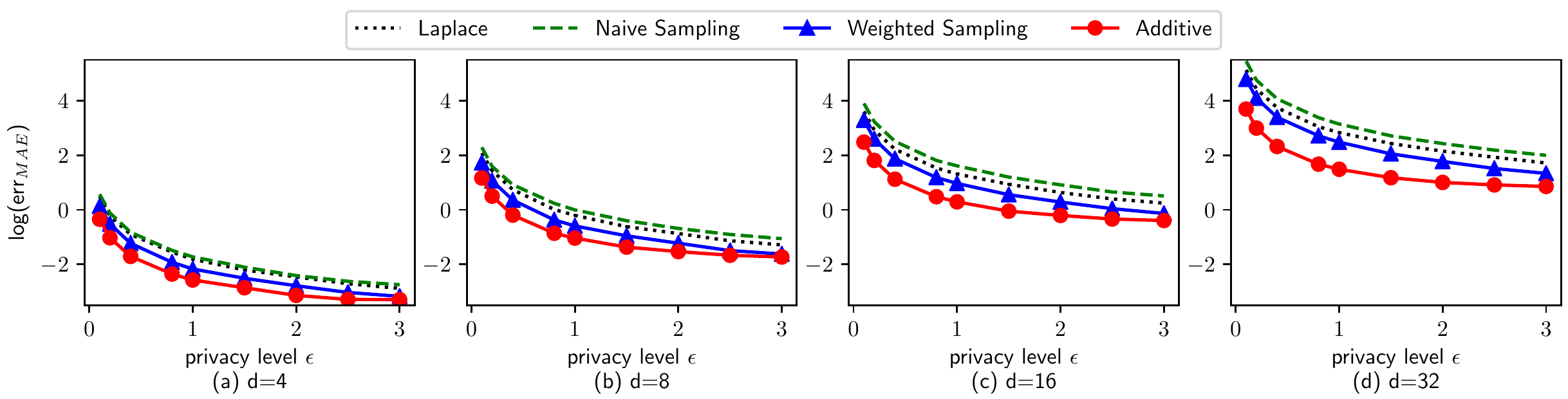}
	\vspace*{-1.5em}
	\caption{Maximum absolute error under Borda rule over $4, 8, 16, 32$ candidates with $10000$ voters.}
	\vspace*{0.5em}
	\label{fig:mae481632borda}
\end{figure*}

\begin{figure*}
	\centering
	\includegraphics[width=175mm]{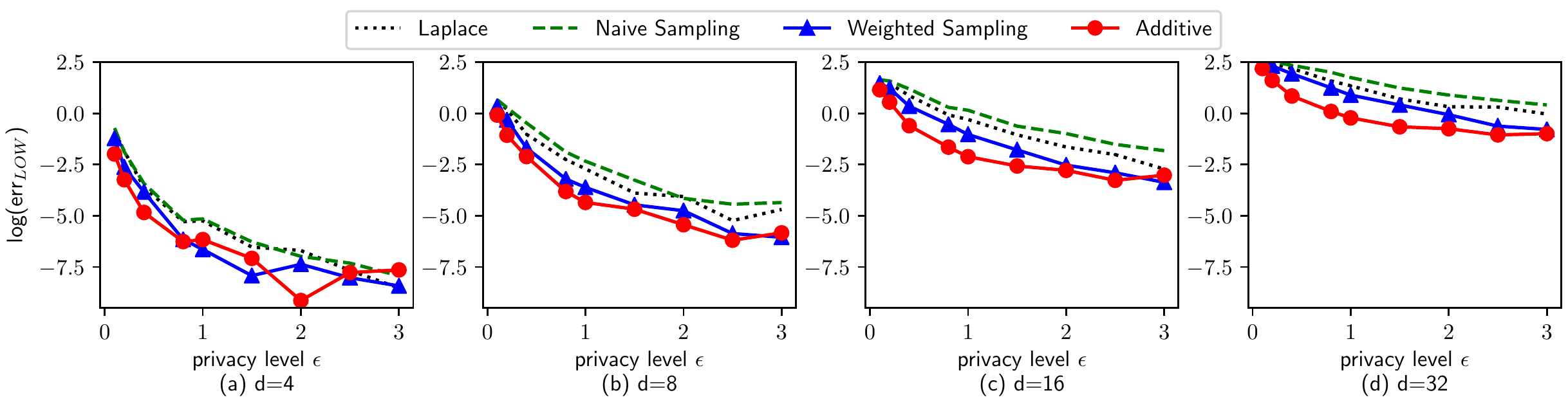}
	\vspace*{-1.5em}
	\caption{Loss of winner error under Borda rule over $4, 8, 16, 32$ candidates with $10000$ voters.}
	\vspace*{0.5em}
	\label{fig:low481632borda}
\end{figure*}

\begin{figure*}
	\centering
	\includegraphics[width=170mm]{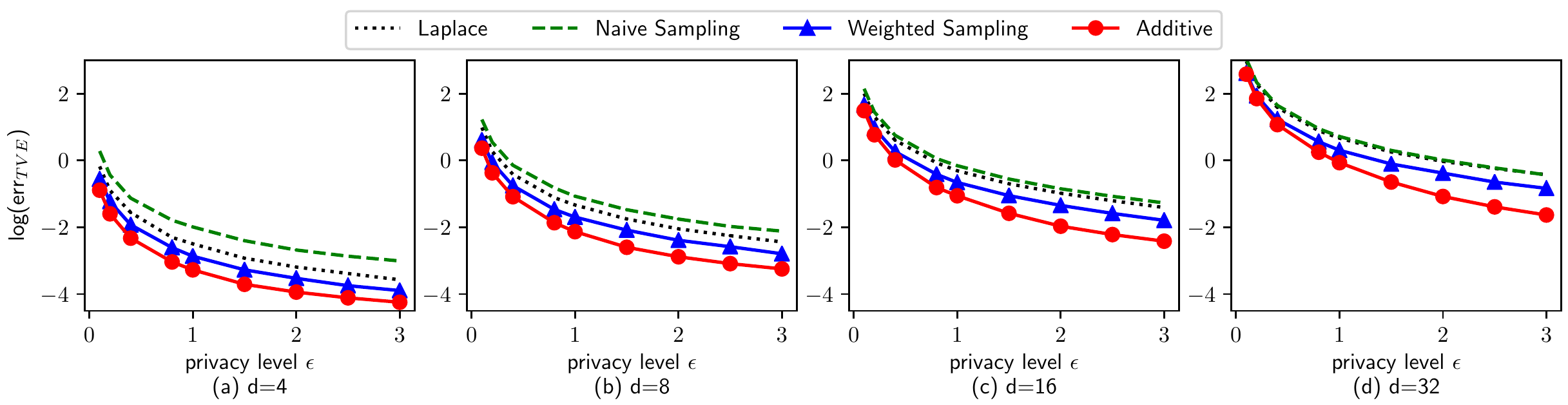}
	\vspace*{-1.5em}
	\caption{Total variation error under Nauru rule over $4, 8, 16, 32$ candidates with $10000$ voters.}
	\vspace*{-0.5em}
	\label{fig:tve481632nauru}
\end{figure*}

\begin{figure*}
	\centering
	\includegraphics[width=175mm]{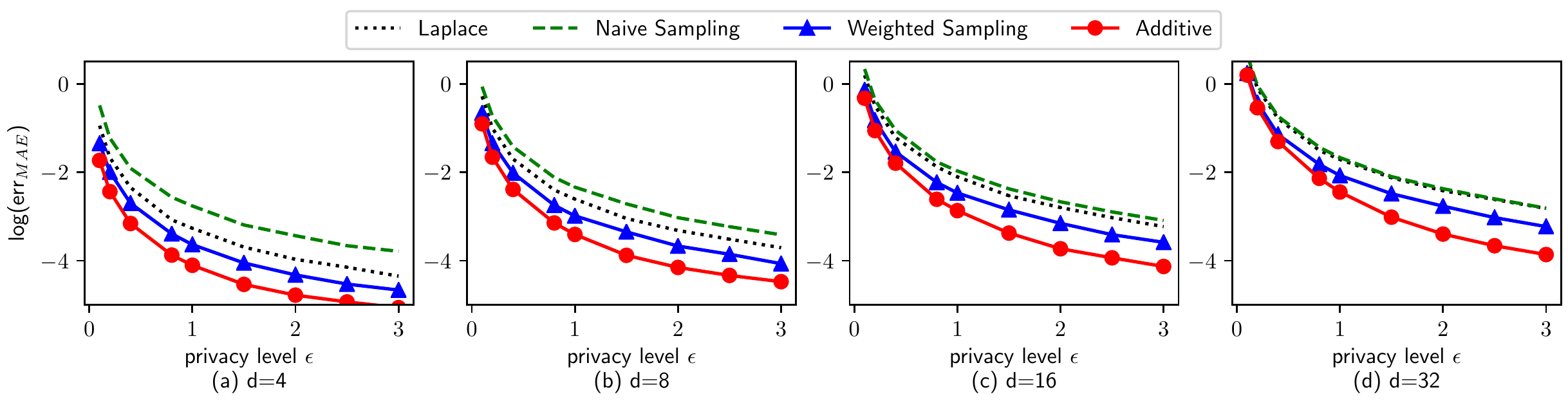}
	\vspace*{-1.5em}
	\caption{Maximum absolute error under Nauru rule over $4, 8, 16, 32$ candidates with $10000$ voters.}
	\vspace*{0.5em}
	\label{fig:mae481632nauru}
\end{figure*}

\begin{figure*}
	\centering
	\includegraphics[width=170mm]{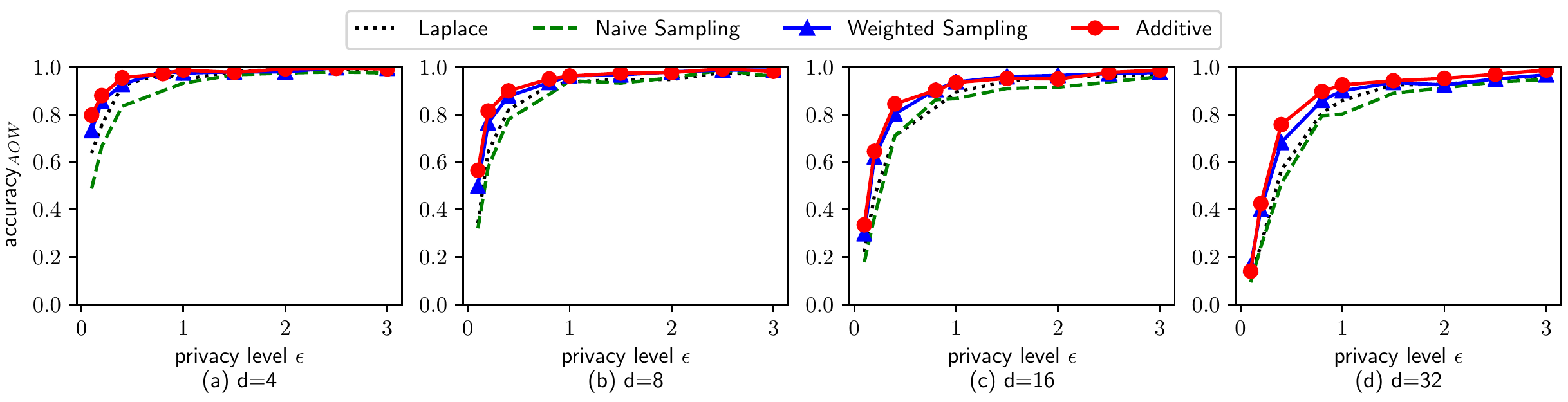}
	\vspace*{-1.5em}
	\caption{Accuracy of winner under Nauru rule over $4, 8, 16, 32$ candidates with $10000$ voters.}
	\vspace*{-0.5em}
	\label{fig:accuracy481632nauru}
\end{figure*}

\begin{figure*}
	\centering
	\includegraphics[width=175mm]{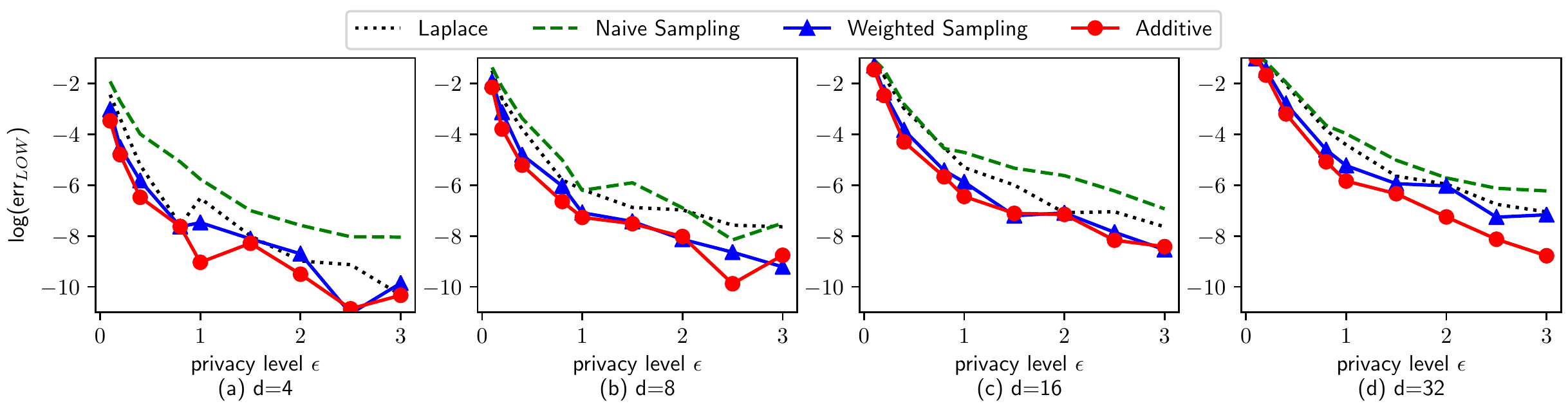}
	\vspace*{-1.5em}
	\caption{Loss of winner error under Nauru rule over $4, 8, 16, 32$ candidates with $10000$ voters.}
	\vspace*{0.5em}
	\label{fig:low481632nauru}
\end{figure*}

\subsection{Additional experimental results of data amplification attacks}\label{result:na}
Results of maximum absolute error, accuracy of winner and loss of winner error under the Nauru rule with varying number of adversarial votes are demonstrated in Figures  \ref{fig:maebordana}, \ref{fig:accuracybordana} and \ref{fig:lowbordana} respectively.
The total variation error, maximum absolute error, accuracy of winner and loss of winner error results under the Nauru rule with varying number of adversarial votes are demonstrated in Figures \ref{fig:tvenauruna}, \ref{fig:maenauruna}, \ref{fig:accuracynauruna} and \ref{fig:lownauruna} respectively.
\begin{figure*}[h]
	\centering
	\includegraphics[width=170mm]{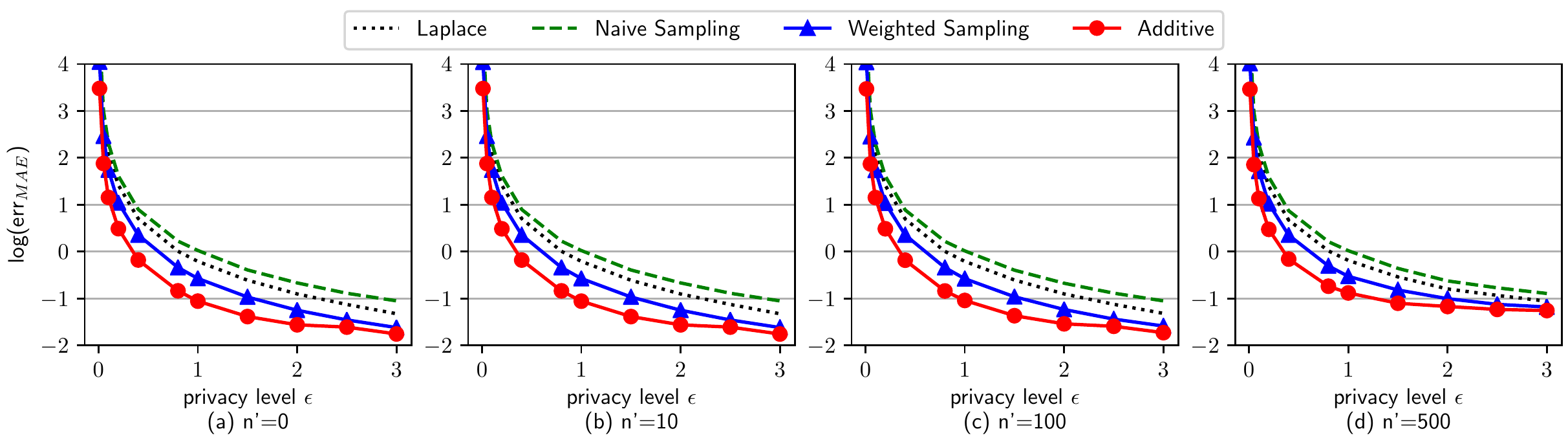}
	\vspace*{-1.5em}
	\caption{Maximum absolute error under Borda rule  with $10000$ honest voters and $n'=0, 10, 100, 500$ adversarial votes}
	\vspace*{-0.5em}
	\label{fig:maebordana}
\end{figure*}

\begin{figure*}[h]
	\centering
	\includegraphics[width=170mm]{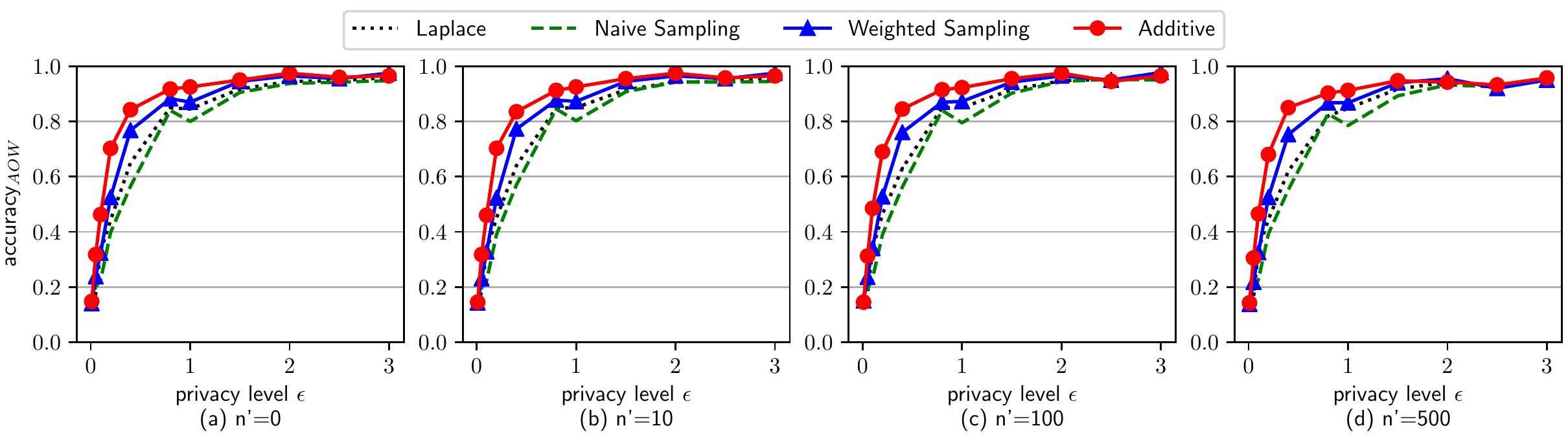}
	\vspace*{-1.5em}
	\caption{Accuracy of winner under Borda rule  with $10000$ honest voters and $n'=0, 10, 100, 500$ adversarial votes.}
	\vspace*{-0.5em}
	\label{fig:accuracybordana}
\end{figure*}

\begin{figure*}[h]
	\centering
	\includegraphics[width=170mm]{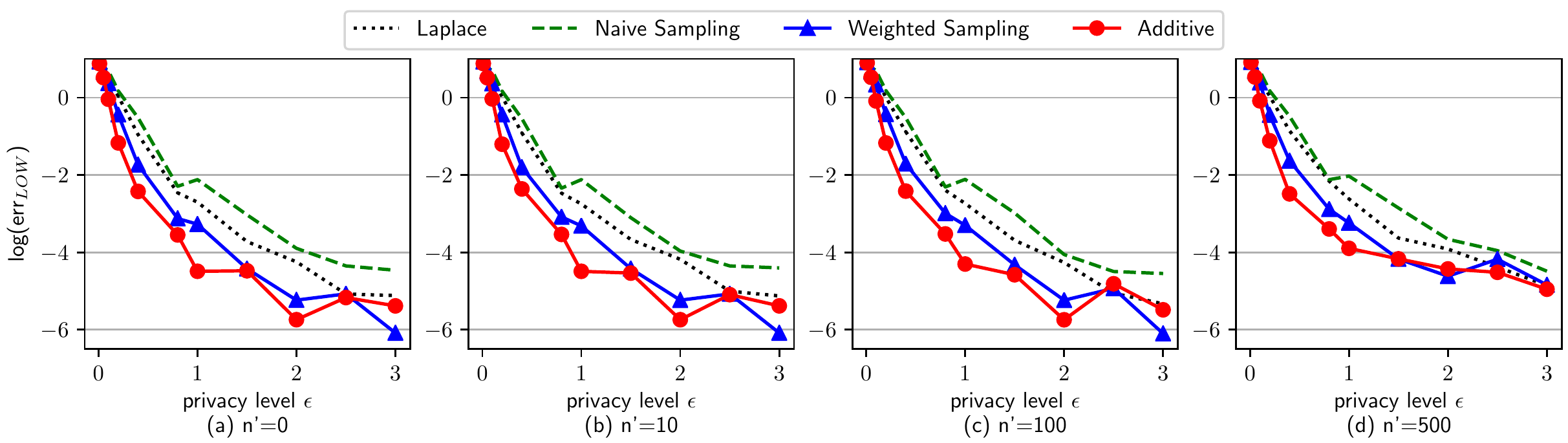}
	\vspace*{-1.5em}
	\caption{Loss of winner error under Borda rule  with $10000$ honest voters and $n'=0, 10, 100, 500$ adversarial votes}
	\vspace*{-0.5em}
	\label{fig:lowbordana}
\end{figure*}


\begin{figure*}[h]
	\centering
	\includegraphics[width=170mm]{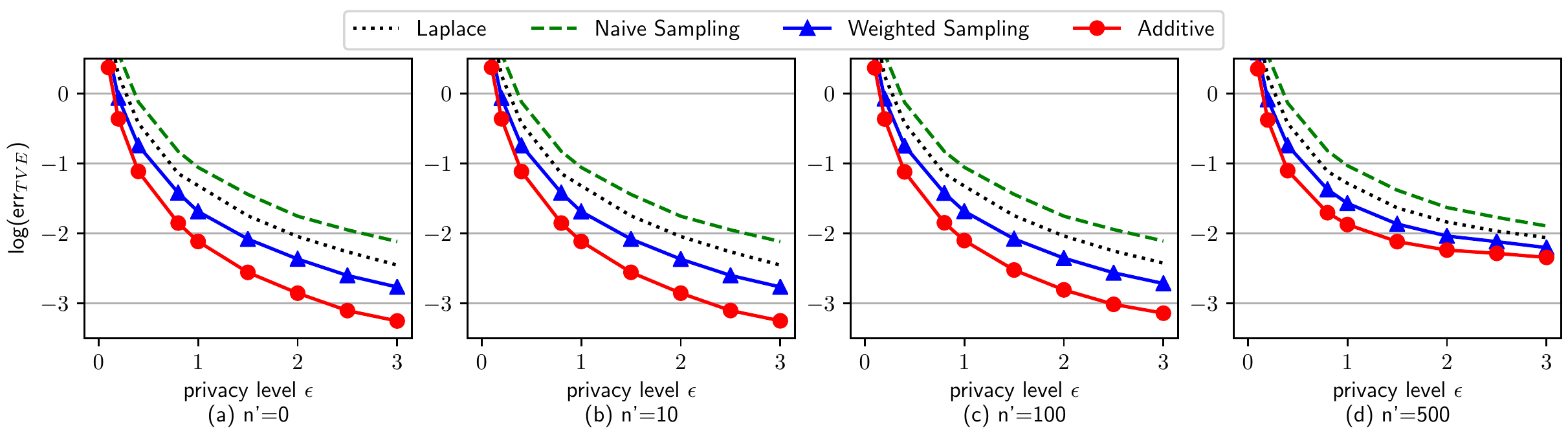}
	\vspace*{-1.5em}
	\caption{Total variation error under Nauru rule with $10000$ honest voters and $n'=0, 10, 100, 500$ adversarial votes.}
	\vspace*{-0.5em}
	\label{fig:tvenauruna}
\end{figure*}

\begin{figure*}[h]
	\centering
	\includegraphics[width=170mm]{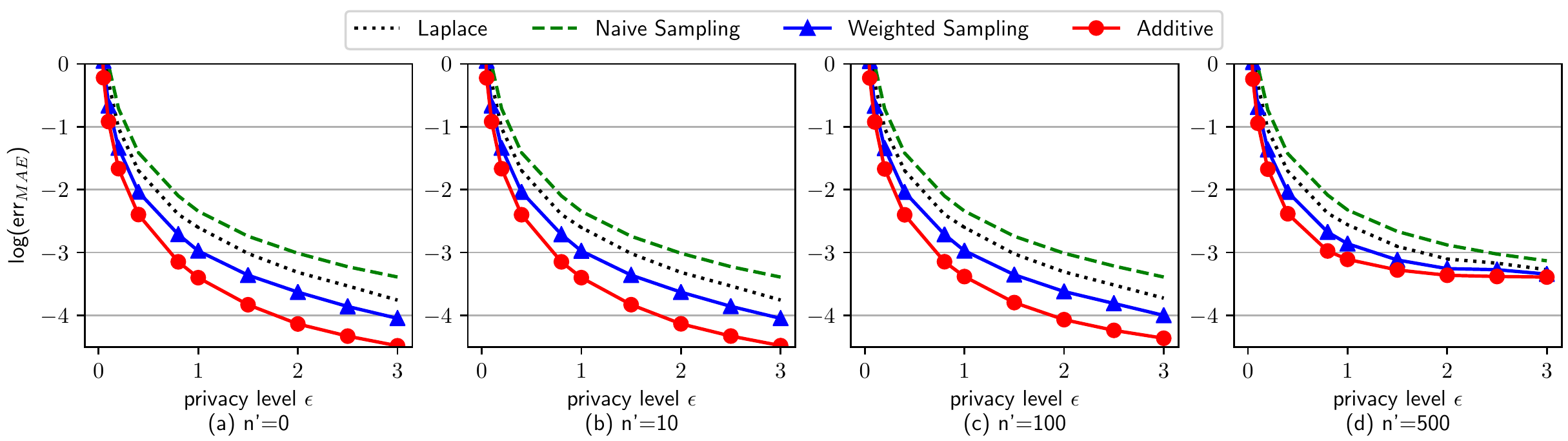}
	\vspace*{-1.5em}
	\caption{Maximum absolute error under Nauru rule with $10000$ honest voters and $n'=0, 10, 100, 500$ adversarial votes}
	\vspace*{-0.5em}
	\label{fig:maenauruna}
\end{figure*}

\begin{figure*}[h]
	\centering
	\includegraphics[width=170mm]{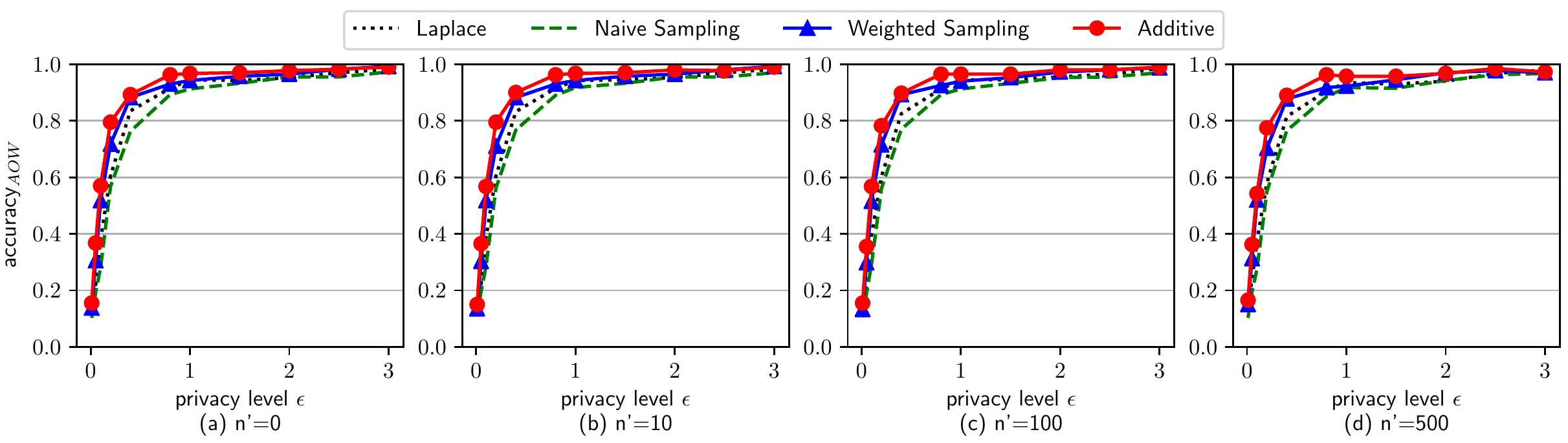}
	\vspace*{-1.5em}
	\caption{Accuracy of winner under Nauru rule with $10000$ honest voters and $n'=0, 10, 100, 500$ adversarial votes.}
	\vspace*{-0.5em}
	\label{fig:accuracynauruna}
\end{figure*}

\begin{figure*}[h]
	\centering
	\includegraphics[width=170mm]{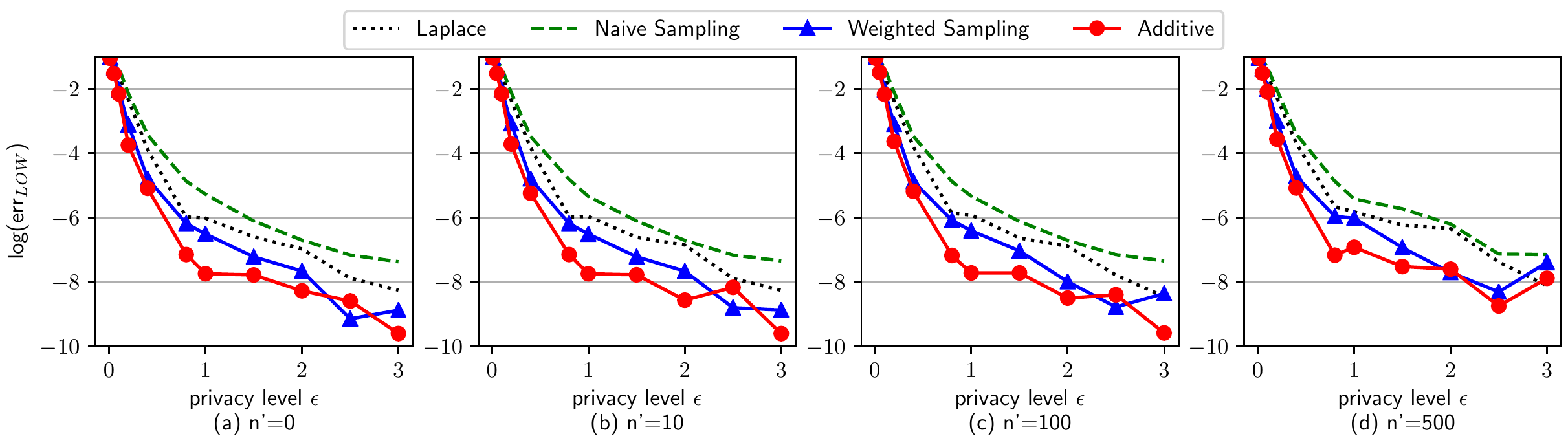}
	\vspace*{-1.5em}
	\caption{Loss of winner error under Nauru rule  with $10000$ honest voters and $n'=0, 10, 100, 500$ adversarial votes}
	\vspace*{-0.5em}
	\label{fig:lownauruna}
\end{figure*}


\subsection{Additional experimental results of view disguise attacks}\label{result:nd}
Results of maximum absolute error, accuracy of winner and loss of winner error under the Nauru rule with varying number of adversarial private views are demonstrated in Figures \ref{fig:maebordand},\ref{fig:accuracybordand}  and \ref{fig:lowbordand} respectively.
The total variation error, maximum absolute error, accuracy of winner and loss of winner error results under the Nauru rule with varying number of adversarial votes are demonstrated in Figures \ref{fig:tvenaurund}, \ref{fig:maenaurund}, \ref{fig:accuracynaurund} and \ref{fig:lownaurund} respectively.

\begin{figure*}[h]
	\centering
	\includegraphics[width=170mm]{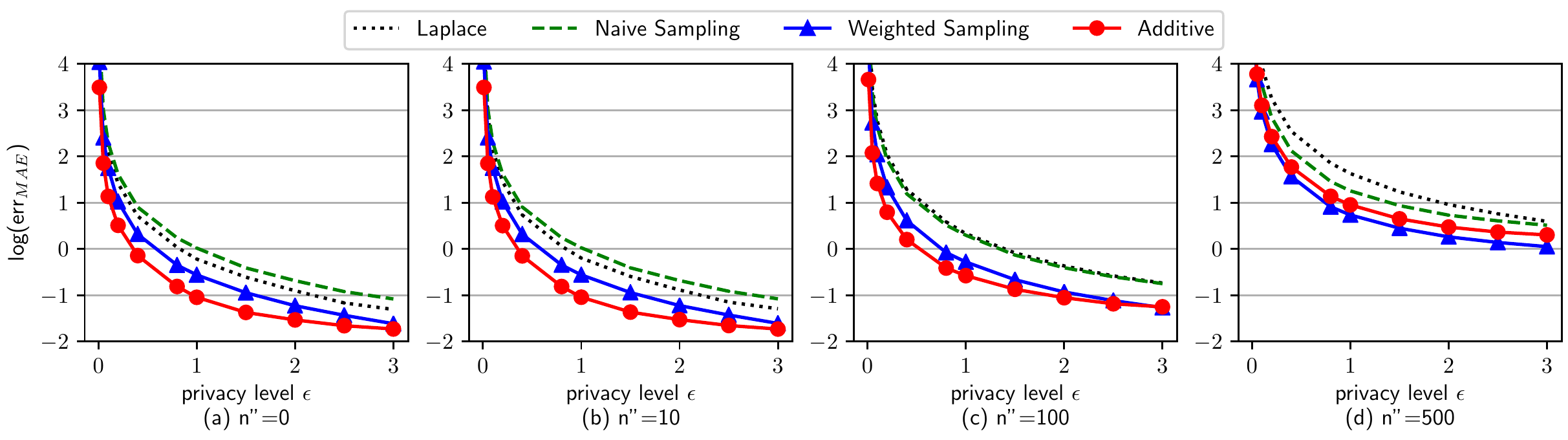}
	\vspace*{-1.5em}
	\caption{Maximum absolute error under Borda rule with $10000$ honest voters and $n'=0, 10, 100, 500$ adversarial private views.}
	\vspace*{-0.5em}
	\label{fig:maebordand}
\end{figure*}

\begin{figure*}
	\centering
	\includegraphics[width=170mm]{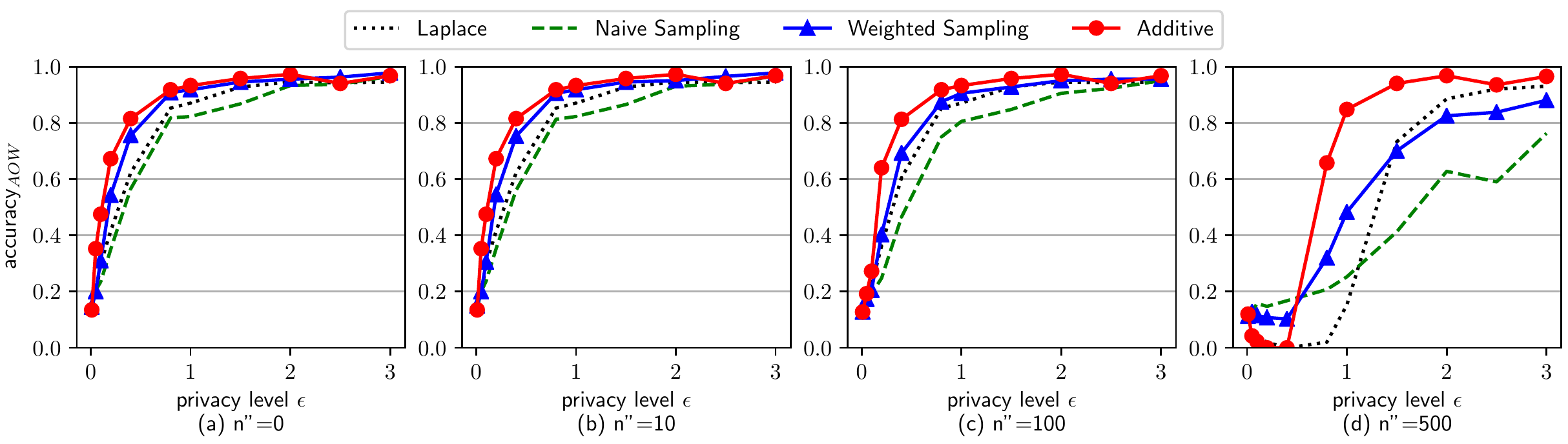}
	\vspace*{-1.5em}
	\caption{Accuracy of winner under Borda rule with $10000$ honest voters and $n'=0, 10, 100, 500$ adversarial private views.}
	\vspace*{-0.5em}
	\label{fig:accuracybordand}
\end{figure*}

\begin{figure*}[h]
	\centering
	\includegraphics[width=170mm]{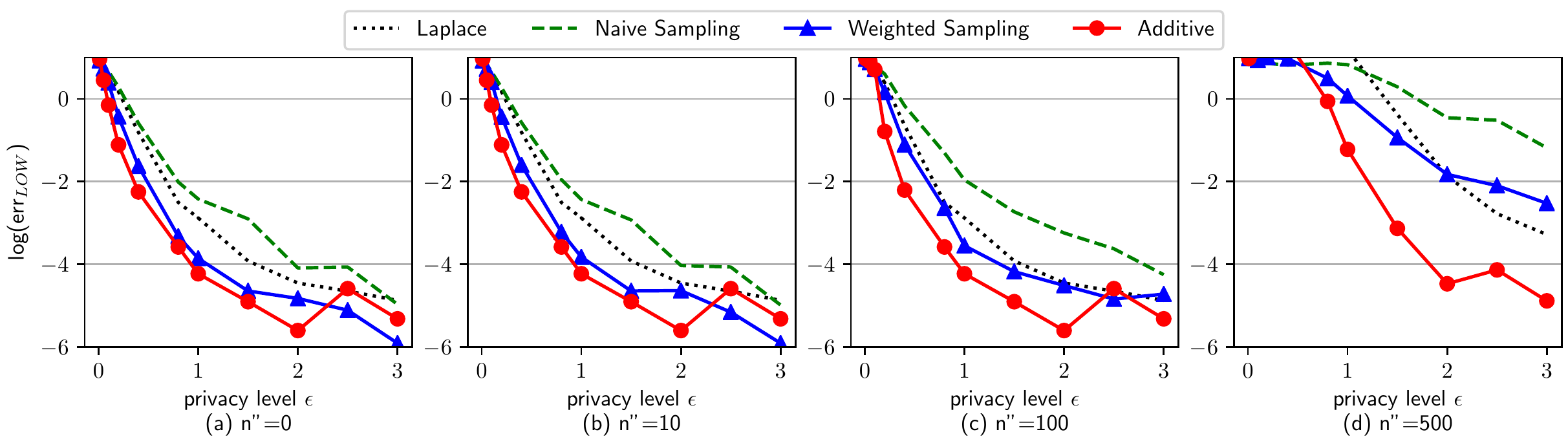}
	\vspace*{-1.5em}
	\caption{Loss of winner error under Borda rule  with $10000$ honest voters and $n'=0, 10, 100, 500$ adversarial private views.}
	\vspace*{-0.5em}
	\label{fig:lowbordand}
\end{figure*}


\begin{figure*}[h]
	\centering
	\includegraphics[width=170mm]{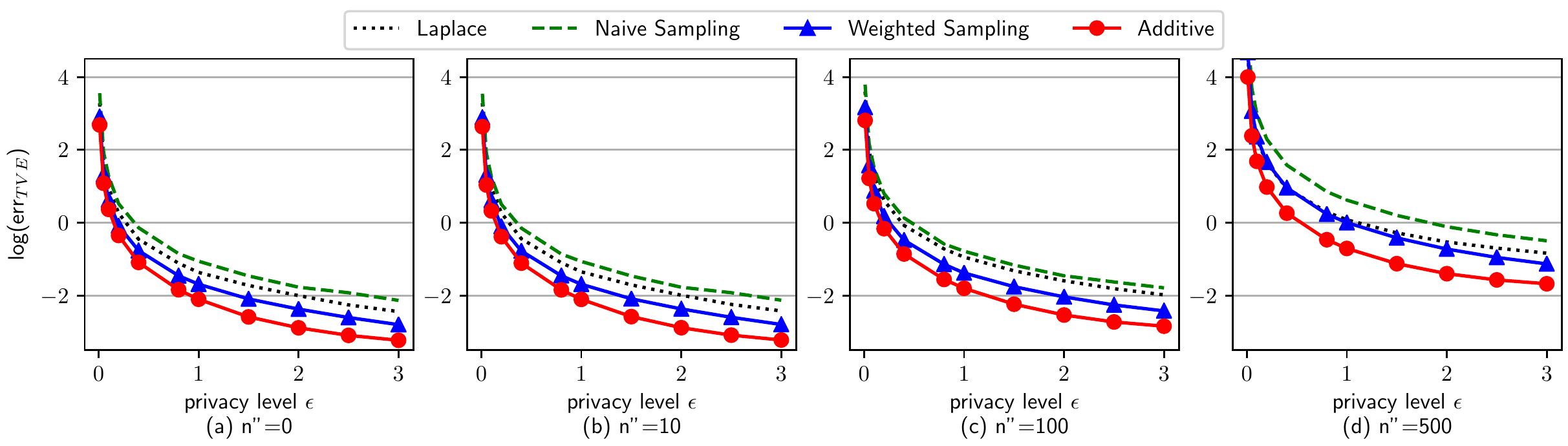}
	\vspace*{-1.5em}
	\caption{Total variation error under Nauru rule with $10000$ honest voters and $n'=0, 10, 100, 500$ adversarial private views.}
	\vspace*{-0.5em}
	\label{fig:tvenaurund}
\end{figure*}

\begin{figure*}[h]
	\centering
	\includegraphics[width=170mm]{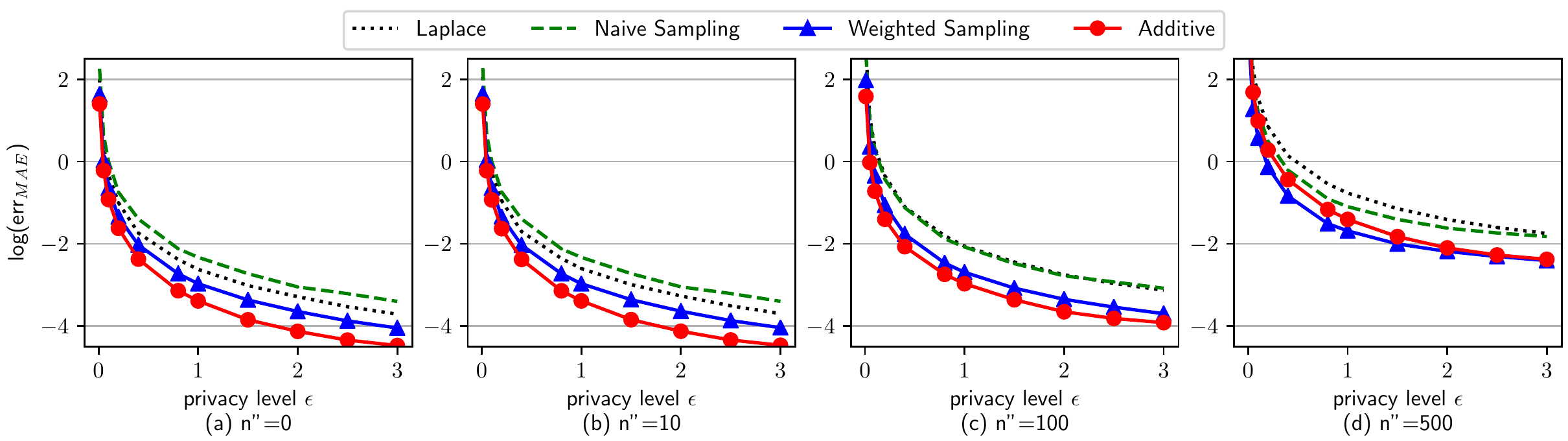}
	\vspace*{-1.5em}
	\caption{Maximum absolute error under Nauru rule with $10000$ honest voters and $n'=0, 10, 100, 500$ adversarial private views.}
	\vspace*{-0.5em}
	\label{fig:maenaurund}
\end{figure*}

\begin{figure*}
	\centering
	\includegraphics[width=170mm]{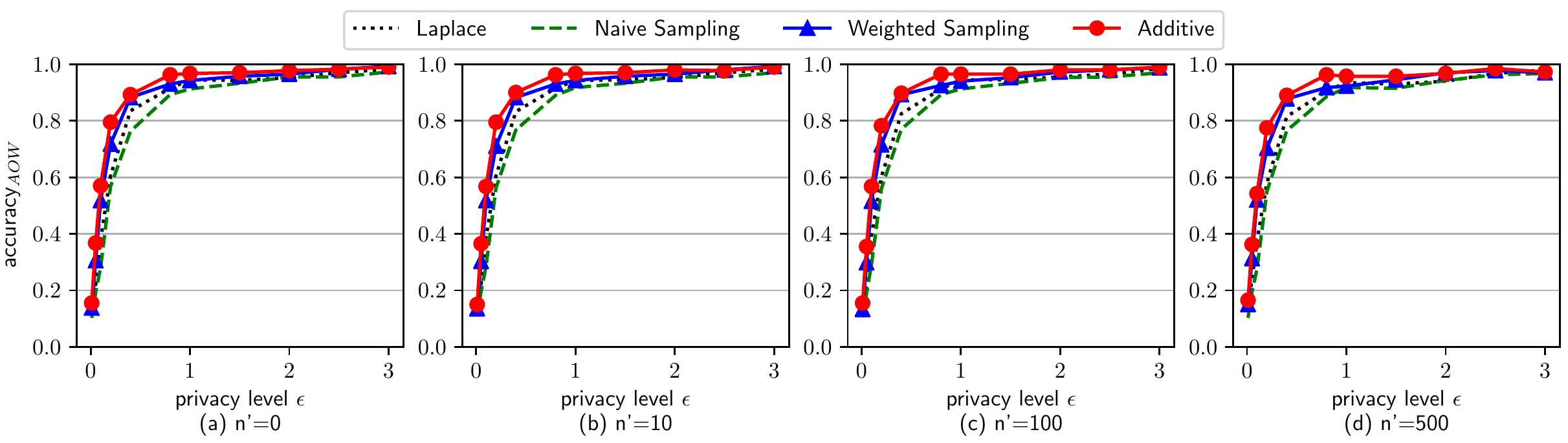}
	\vspace*{-1.5em}
	\caption{Accuracy of winner under Nauru rule with $10000$ honest voters and $n'=0, 10, 100, 500$ adversarial private views.}
	\vspace*{-0.5em}
	\label{fig:accuracynaurund}
\end{figure*}

\begin{figure*}[h]
	\centering
	\includegraphics[width=170mm]{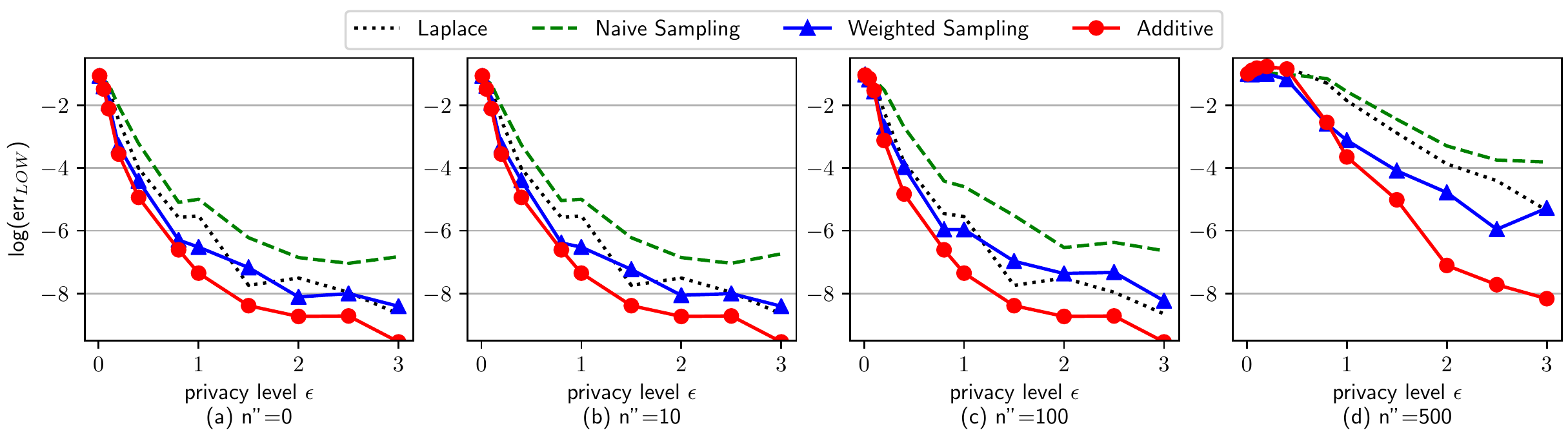}
	\vspace*{-1.5em}
	\caption{Loss of winner error under Nauru rule  with $10000$ honest voters and $n'=0, 10, 100, 500$ adversarial private views.}
	\vspace*{-0.5em}
	\label{fig:lownaurund}
\end{figure*}



\end{document}